\documentclass[journal]{IEEEtran}
\usepackage{amssymb,amsthm,bm}
\usepackage[tbtags]{amsmath}
\usepackage{datetime}
\usepackage{graphicx}
\graphicspath{{figures-pdf/}}
\usepackage[caption=false,font=footnotesize,labelfont=rm,textfont=rm]{subfig}

\usepackage{cite} 
\usepackage{url} 
\usepackage[aboveskip=1pt]{subcaption}
\usepackage{diagbox}
\usepackage{multirow, bigstrut}
\usepackage{arydshln} 
\usepackage[mathlines]{lineno}

\usepackage{verbatim} 
\usepackage{comment} 

\usepackage{hyperref}
\hypersetup{linktocpage=true, pdfborderstyle={/S/S/W 1}, bookmarksopen=true,
colorlinks,linkcolor=blue,anchorcolor=blue,citecolor=blue}

\newlength\OneImW
\newlength\BigOneImW
\newlength\twofigwidth
\newtheorem{proposition}{Proposition}

\newtheorem{theorem}{Theorem}
\newtheorem{lemma}{Lemma}

\newcommand{\ord}{\mathrm{ord}}
\newcommand{\Zp}[1]{\mathbb{Z}_{p^{#1}}}
\newcommand{\Zt}[1]{\mathbb{Z}_{2^{#1}}}

\newcommand{\vvert}[0]{{\,\vert \,}}

\newcommand{\newItemizeWidth}{%
	\setlength{\labelwidth}{\widthof{\textbullet}}%
	\setlength{\labelsep}{3pt}%
	\setlength{\IEEElabelindent}{4pt}%
	\IEEEiedlabeljustifyl
}

\begin{document}

\title{Graph Structure of Chebyshev Permutation Polynomials over Ring $\Zp{k}$}
\author{Chengqing Li, Xiaoxiong Lu, Kai Tan, and Guanrong Chen, 
\IEEEmembership{Life Fellow,~IEEE}
    \thanks{This work was supported by the National Natural Science Foundation of China (no.~92267102). \textit{(Corresponding author: Xiaoxiong Lu.)}
    }
    
    \thanks{C. Li is with the School of Computer Science, Xiangtan University, Xiangtan 411105, Hunan, China (\protect\url{DrChengqingLi@gmail.com}).}

    \thanks{X. Lu is with the School of Mathematics and Computational Science, Xiangtan University, Xiangtan 411105, China (\protect\url{xiaoxionglu1@foxmail.com}).}
    
    \thanks{K. Tan is with the Key Laboratory of Intelligent Computing \& Information Processing of Ministry of Education, Xiangtan University, Xiangtan 411105, Hunan, China }
    
    \thanks{G. Chen is with the Department of Electrical Engineering, City University of Hong Kong, Hong Kong SAR, China.}
}
    \markboth{IEEE Transactions}{Li \MakeLowercase{\textit{et al.}}}
    \IEEEpubid{\begin{minipage}{\textwidth}\ \\[12pt]\centering
    \\[2\baselineskip]
    0018-9448 \copyright 2024 IEEE. Personal use is permitted, but republication/redistribution requires IEEE permission.\\
    See http://www.ieee.org/publications\_standards/publications/rights/index.html for more information.\\
    \end{minipage}
}

\maketitle

\begin{abstract}
Understanding the underlying graph structure of a nonlinear map over a particular domain
is essential in evaluating its potential for real applications.
In this paper, we investigate the structure of the associated \textit{functional graph} of Chebyshev permutation polynomials over a ring $\Zp{k}$, with $p$ being a prime number greater than three, where every number in the ring is considered as a vertex and the existing mapping relation between two vertices is regarded as a directed edge.
Based on some new properties of Chebyshev polynomials and their derivatives, we disclose how the basic structure of the functional graph evolves with respect to parameter $k$.
First, we present a complete and explicit form of the length of a path starting from any given vertex.
Then, we show that the strong patterns of the functional graph that the number of cycles of any given length always remains constant as $k$ increases.
Moreover, we rigorously prove the rules on the elegant structure of the functional graph and verify them experimentally.
Our results could be useful for studying the emergence of the complexity of a nonlinear map in digital computers and security analysis of its cryptographic applications.
\end{abstract}

\section{Introduction}

\IEEEPARstart{C}{hebyshev} polynomial, named after the Russian mathematician Pafnuty L. Chebyshev~\cite{Chebyshev:life:JAT99},
find applications in diverse fields, such as function approximation~\cite{Lee:cheby-approxi:TDSC22}, 
pseudorandom number generator~\cite{Liu:Chebyrandom:2011}, spread-spectrum sequence \cite{Chencc:chebyshev:TCASI2001}, 
authentication \cite{chatterjee2018DSC,Tomar:auth:TMC2024}, key exchange protocol \cite{Dariush:Chebyshevprivacy:TII2020}, and
privacy protection~\cite{Zhang:Privacy:TII:2022}.
The Chebyshev polynomials of the first kind are defined 
by the recurrence relation:
\begin{equation}\label{eq:recursive}
T_n(x)=2xT_{n-1}(x)-T_{n-2}(x),
\end{equation}
where $T_0(x)=1$, $T_1(x)=x$, and $n$ is the degree of the
polynomial. The Chebyshev polynomials of the second kind 
follow a similar recurrence relation:
\(
U_{n+1}(x)=2xU_n(x)-U_{n-1}(x), 
 \)
with initial conditions $U_0(x)=1$ and $U_1(x)=2x$.
Moreover, one may define Chebyshev polynomials as solutions to the Pell equation~\cite{Rivlin:chebyshev:1990}, which yields
\begin{equation}\label{eq:s:sunx}
\left(x+\sqrt{x^2-1}\right)^n=T_n(x)+U_{n-1}(x)\sqrt{x^2-1}.
\end{equation}
When the variable $x$ is defined over real numbers, the recurrence relation in
\eqref{eq:recursive} simplifies to
$T_n(x)=\cos(n\arccos(x))$, resulting in the explicit expression for the degree:
$n\in\{\frac{ \pm \arccos(T_n(x))+2i\pi}{\arccos(x)}\vvert i \text{~is an integer}\}$.
However, when defined over other finite algebraic fields, such explicit expressions typically do not exist 
due to quantization and algorithmic errors in calculating trigonometric functions over finite-precision domains~\cite{Bergamo:PKIchebyshev:CASI2005,Liao:ITC:2010}.

\IEEEpubidadjcol

The semigroup property of the Chebyshev polynomials (See definition in Sec.~\ref{ssec:pre}) has attracted considerable attention from the field of cryptography.
These polynomials are employed as an innovative approach to design public-key encryption algorithms \cite{Kocarev:2003, Bergamo:PKIchebyshev:CASI2005, Cheong:PKIchebyshev:CASII2007, Kocarev:2005}, with the first publication on the polynomials over the field of all rational and irrational numbers in 2003~\cite{Kocarev:2003}.
However, since the degree of a Chebyshev polynomial can be explicitly expressed over the real number field, recent analyses have uncovered security flaws in the algorithm designed initially in \cite{Kocarev:2003}, as reported in~\cite{Bergamo:PKIchebyshev:CASI2005, Cheong:PKIchebyshev:CASII2007}. 
Subsequently, the numerical domain was changed from the real number field to ring $\Zp{k}$~\cite{Kocarev:2005}, a ring of residue classes modulo $p^k$, where $k$ is a positive integer and $p$ is a prime number.
The security of the associated algorithms depends on the difficulty of deriving the degree of the Chebyshev polynomial used, which satisfies
$y\equiv T_n(x)\pmod{p^k}$ for given $x$, $y$, and $k$. Recent cryptanalysis
has shown that the public-key encryption algorithm is insecure when the parameters are improperly selected~\cite{Liao:ITC:2010, Chen:chebyZN:IS2011, Yoshioka:ChebyshevPk:TCAS2:2020, Yoshioka:Chebyshev2k:TCAS2:2016, Yoshioka:ChebyshevPk:TCAS2:2018}.
Nevertheless, a complete security analysis of cryptographic applications of Chebyshev polynomials over $\Zp{k}$ necessitates disclosure of its graph structure.

There are two types of sequences produced by Chebyshev polynomials over $\Zp{k}$: Chebyshev degree sequence
$\{T_i(x)\bmod p^k\}_{i\ge 0}$~\cite{Liao:ITC:2010, Chen:chebyZN:IS2011, Yoshioka:ChebyshevPk:TCAS2:2020} and Chebyshev integer
sequence~$\{T_n^i(x)\bmod p^k\}_{i\ge 0}$~\cite{Umeno:permute:2018, Yoshioka:Chebyshev2k:TCAS2:2016, Yoshioka:ChebyshevPk:TCAS2:2018}, where
$T_n^i(x)$ is the $i$-th iteration of $T_n(x)$.
The former is essentially a homogeneous linear recurring sequence, characterized by its periodicity for any set of parameters and initial states~\cite{Rudolf:Introfinite:1994}.
The period distributions of Chebyshev degree sequences over prime field $\mathbb{F}_p$~\cite{Liao:ITC:2010} and integer ring $\mathbb{Z}_q$~\cite{Chen:chebyZN:IS2011} have been systematically analyzed by the generating function method, showing that the security of the public-key encryption algorithm based on a Chebyshev polynomial depends on the parameter $p$ and $q$, where $q$ is composite.
Combining the variation rule of the period of a Chebyshev degree sequence with increasing parameter $k$, a polynomial-time algorithm was designed in~\cite{Yoshioka:ChebyshevPk:TCAS2:2020} to determine the equivalence class of the degree of a Chebyshev polynomial modulo $p^k$ when $p$ is relatively small. Chebyshev integer sequence is defined by a nonlinear congruence iterative sequence.
It is periodic if and only if the Chebyshev polynomial acts as a permutation polynomial over $\Zp{k}$.
The periods of the Chebyshev integer sequences, generated by iterating Chebyshev permutation polynomials over rings $\Zt{k}$ and $\Zp{k}$, were investigated in \cite{Umeno:permute:2018, Yoshioka:Chebyshev2k:TCAS2:2016, Yoshioka:ChebyshevPk:TCAS2:2018}.
Moreover, an algorithm with time complexity $O(k)$ was developed for determining the degree of the Chebyshev polynomial modulo $2^k$~\cite{Yoshioka:Chebyshev2k:TCAS2:2016}. 
It was reported in \cite{Yoshioka:ChebyshevPk:TCAS2:2018} that a key exchange protocol based on a Chebyshev polynomial modulo $p^k$ may be vulnerable to brute-force attacks under specific parameters. 
While the period of Chebyshev degree sequences across various domains can be analyzed using algebraic tools typical for linear recurring sequences, the nonlinear complexity of Chebyshev integer sequences renders their period analysis more challenging.

The graph structure of a nonlinear system over different domains can be used to evaluate its dynamic complexity and randomness \cite{chen:PRNS:TC22, cqli:network:TCASI2019, cqli:Cat:TC22, licq:Logistic:IJBC2023, Chen:cat:TIT2012, chenf:cat2:TIT13}.
In~\cite{cqli:network:TCASI2019, cqli:Cat:TC22, Li:ISIT2024}, 
the evolution rules for the graph structure of various chaotic maps, including the Logistic, Tent, and Cat maps, on finite-precision computers, were elucidated.
We thoroughly detailed the graph structure and associated evolution rules of the generalized Cat map over any binary arithmetic domain.
The periodicity properties of sequences generated by the Logistic and Cat maps over different rings were discussed in \cite{licq:Logistic:IJBC2023, Chen:cat:TIT2012, chenf:cat2:TIT13}.
The theoretical basis for the graph structure of a linear feedback shift register with an arbitrary characteristic polynomial over the finite field $\mathbb{F}_{p^k}$ was established based on cyclotomic classes and decimation sequences~\cite{Chang:cycleLFSR:CC18}.
Additionally, the graph structure of a class of cascaded feedback registers was determined by solving a system of linear equations~\cite{gongg:cycle:TIT19}.
Using a specialized tree associated with $v$-series, the functional graphs of R{\'e}dei functions~\cite{Daniel:Redei:DM2015}, Chebyshev polynomials~\cite{Daniel:Chebyshev:2018}, and general linear maps~\cite{Daniel:lineargraph:DCC2019} over the finite field $\mathbb{F}_{p^k}$ were analyzed concerning the number of connected components, cycle lengths, and the average pre-period (transient) lengths.
Moreover, a comprehensive review was presented in \cite{Li:neteork:ISCAS2019} on studying the dynamics of digitized nonlinear maps via the functional graph.

Although the graph structure of Chebyshev polynomials over finite fields $\mathbb{F}_{p^k}$ and ring $\Zt{k}$ has been extensively investigated in the past decade, as reported in \cite{Umeno:permute:2018, Yoshioka:Chebyshev2k:TCAS2:2016, Yoshioka:ChebyshevPk:TCAS2:2018,Daniel:Chebyshev:2018, gassert:chebyshevaction:DM2014,Yoshioka:ISIT2023}, its structure over the ring $\Zp{k}$ remains an open question, potentially due to the high computational complexity associated with higher-degree Chebyshev polynomials.
Notably, D. Yoshioka provided the periodicity of the Chebyshev integer sequence 
$\{T_n^i(x)\}_{i\ge 0}$ generated from partial initial states over the ring $\Zp{k}$
under specific assumptions~\cite{Yoshioka:ChebyshevPk:TCAS2:2018}. However, the exact explicit expression of the
periodicity remains a subject of debate, as highlighted at the beginning of Sec.~\ref{subsec:cheby:integer}.
The primary contributions of this paper are as follows:
\begin{itemize}[\newItemizeWidth]
\item We rigorously prove some new properties of Chebyshev polynomials and
their derivative over the ring $\Zp{k}$.

\item We refine and enhance the periodicity analysis of Chebyshev integer sequences over the ring $\Zp{k}$ as presented in \cite{Yoshioka:ChebyshevPk:TCAS2:2018}.

\item We disclose the graph structure of Chebyshev permutation polynomials over the ring $\Zp{k}$ and especially how it varies with the parameter $k$, enriching the research on polynomial congruences with prime power moduli with analytic number theory.

\item We develop an effective parameter selection strategy to thwart brute-force attacks on a key exchange protocol based on Chebyshev polynomials.
\end{itemize}
This paper also offers new perspectives for constructing complete permutation polynomials investigated in \cite{Weng:notePP:TIT2008,Sun:newCPP:TIT2021,Wu:CPP:TIT2023}.

The rest of the paper is organized as follows. Section~\ref{Sec:Chebypre}
derives the period of the Chebyshev integer sequence generated from any initial state over the ring $\Zp{k}$.
The graph structure of Chebyshev permutation polynomials over the ring is completely disclosed in Sec.~\ref{sec:network}.
Section~\ref{sec:application} briefly discusses some cryptographic applications of the studied polynomials.
The last section concludes the paper.

\section{Some Properties of Chebyshev Polynomial and Chebyshev Integer Sequence}
\label{Sec:Chebypre}

In this section, we revisit the established results concerning the Chebyshev polynomial and introduce novel properties of its $m$-order derivative. These properties facilitate the analysis of the explicit expression of the least period of Chebyshev integer sequence originating from 
any initial state.

\subsection{Preliminary}
\label{ssec:pre}

To assist in the subsequent discussion, Table~\ref{tab:list:symbol} provides a list of notation that recur throughout this paper, excluding those that appear only in specific contexts.

\begin{table*}[!htb]
\caption{Nomenclature.}
\centering
\setlength{\tabcolsep}{2pt}
    \begin{tabular}{ll|ll}
    \hline
    $\mathbb{Z}$ & The set of all integers.  & $\mathbb{N}^+$  & The set of all positive integers.\\ \hline
    $\mathcal{Z}_{p^k}$ & The set of non-negative integers less than $p^k$.    & $T^{(m)}_n(x)$  & The $m$-order derivative of Chebyshev polynomial $T_n(x)$ satisfying \eqref{PermuteCondition}.\\ \hline
    $\Zp{k}$           & The ring of residue classes of the integers modulo $p^k$.      & $T^m_n(x)$ & The composition of Chebyshev polynomial $T_n(x)$ with itself $m$ times.\\ \hline
    $N_x$ & The least period of sequence $\{T_n^m(x)\bmod p\}_{m\ge 0}$. & $v_p(x)$ &The exponent of the largest power of $p$ that divides integer $x$.\\ \hline
    $b\bmod q$ & 
    $b-\lfloor\frac{b}{q}\rfloor$, where $\lfloor\cdot\rfloor$is the floor function.
    & $s_p(x)$  & The sum of the digits in the base-$p$ expansion of $x$.\\ \hline
    $a\mid b$ &The integer $a$ divides the integer $b$.
         & $\ord(x)$  & The least positive number $j$ such that $x^j\equiv 1\pmod p$. \\ \hline
    $\gcd(a,b)$ & The greatest common divisor of the integers $a$ and $b$.
    &
    $a\equiv b\pmod{q}$ & The base $a$ is congruent to the residue $b$ modulo $q$.
    \\\hline
    \end{tabular}
\label{tab:list:symbol}
\end{table*}

For any integer $n > 1$, the Chebyshev polynomial of degree $n$ can be expressed in power series form:
\begin{equation}
T_n(x)=\frac{n}{2} \sum_{m=0}^{\lfloor n/2 \rfloor} (-1)^m \frac{(n-m-1)!}{m! (n-2m)!} (2x)^{n-2m},
\label{cheby:power}
\end{equation}
as detailed in \cite{Lidl:Dick:1993}. One of the most significant properties of Chebyshev polynomials is the semigroup property, articulated as
$T_{n_1}(T_{n_2}(x))=T_{n_2}(T_{n_1}(x))=T_{n_1 \cdot n_2}(x)$
for any two positive integers $n_1$ and $n_2$.
Consequently, the composition of the Chebyshev polynomial $T_n(x)$ with itself $N$ times
is denoted as $T^N_n(x)=T_{n^N}(x)$. This semigroup property of Chebyshev polynomials over $\Zp{k}$ remains valid, namely
$T_{n_1}(T_{n_2}(x)) \equiv T_{n_2}(T_{n_1}(x)) \equiv T_{n_1 \cdot n_2}(x)\pmod{p^k}$.

Given an initial state $s\in\mathcal{Z}_{p^k}$, the $i$-th iteration of $T_n(s)$ over $\Zp{k}$ is
$T^i_n(s)=T_n(T^{i-1}_n(s))\bmod p^k$, 
where $i\geq 1$ and $T^0_n(s)=s$.
The least period of a sequence generated by iterating a Chebyshev polynomial is defined as the minimum positive integer $N$ such that
$T^N_n(s)\equiv s\pmod{p^k}$.
A function that can return to its initial state after a certain number of iterations is referred to as a \textit{permutation function} and exists if and only if it is bijective over its domain.
When $k>1$, Yoshioka demonstrated that a Chebyshev polynomial $T_n(x)$ over $\Zp{k}$ is a \textit{permutation polynomial} if and only if
\begin{equation}\label{PermuteCondition}
\gcd(n, p^3-p)=1,
\end{equation}
indicating that $2\nmid n$ and $3\nmid n$~\cite{Yoshioka:ChebyshevPk:TCAS2:2018}.
Throughout the remainder of this paper, unless otherwise specified, it is assumed that
$n\in \mathbb{N}^+$, $p$ is a prime number greater than 3,
and they satisfy condition~\eqref{PermuteCondition}. 
The graph structure of the Chebyshev polynomial with $p=2$ is examined in \cite{Yoshioka:Chebyshev2k:TCAS2:2016},
while that with $p=3$ requires a distinct analysis, as discussed in \cite{licq:Logistic:IJBC2023}, which addresses the Logistic map over the ring $\mathbb{Z}_{3^k}$.

\subsection{Some Properties of the Derivative of a Chebyshev Polynomial}

Referring to the Taylor expansion of the Chebyshev polynomials $T_n(x)$, one observes that $\frac{T^{(m)}_n(x)}{m!}$ is an integer for any integer $x$. Moreover, according to Legendre's formula (see \cite{Moll:numbersandfunctions:2012}), Lemma~\ref{lemma:pw:Tn'1} describes the number of factors of $p$ contained in the $m$-order derivative of $T_n(x)$ evaluated at $\pm 1$. Note that Lemma~\ref{lemma:pw:Tn'1} does not hold for the case $(m, p)=(2, 3)$. 
In addition, Lemma~\ref{lemma:tn(-1+p):equal} 
establishes an equivalence between the first-order derivatives of $T_n(x)$ at $-1$ and $p-1$ 
provided that these values are congruent to $1$ over $\Zp{w}$, where $w\in \mathbb{N}^+$.
The interrelation between derivatives of three consecutive orders of $T_n(x)$ is discussed in Lemma~\ref{lemma:dTn:x}.
Lemma~\ref{lemma:T'n(x):n} demonstrates that the first derivative of $T_n(x)$ at a state that the polynomial maps onto itself is equal to $\pm n$ over $\Zp{w}$. 
Eliminating the higher-order terms in the Taylor expansion of $T_n(x)$ using these four lemmas, Lemma~\ref{lemma:Tn:0pw} provides a sufficient condition that guarantees 
the first-order derivative of $T_n(x)$ at $s+h\cdot p$
is congruent to that at $s$ over $\Zp{w+1}$.
Additionally, the higher-order derivatives of $T_n(x)$ at $s+h\cdot p$ are congruent to zero over $\Zp{w+2}$, where $s\in\mathcal{Z}_p$ and $h\in\mathbb{Z}$. Furthermore, Proposition~\ref{lemma:Tn(x):w} gives a sufficient condition for extending the fixed point $s$
of $T_n(x)$ over $\Zp{}$ to another fixed point $s+h\cdot p$ over $\Zp{w+1}$.

\begin{lemma}\label{lemma:pw:Tn'1}
If the derivative of the Chebyshev polynomial $T_n(x)$ satisfies
$T'_n(\pm 1)\equiv 1\pmod{p^w}$, then
$\frac{T^{(m)}_n(\pm 1) \cdot p^m}{m!} \equiv 0\pmod{p^{w+2}}$ when \(m\geq 2\), where $w \in \mathbb{N}^+$.
\end{lemma}
\begin{proof}
As shown in \cite[Eq. (1.97)]{Rivlin:chebyshev:1990}, the $m$-order derivative of the Chebyshev polynomial $T_n(x)$ at $\pm 1$ can be expressed as
\begin{equation}\label{eq:t'n1}
T^{(m)}_n(\pm 1)=(\pm1)^{n+m}\frac{\prod_{j=0}^{m-1}(n^2-j^2)}{(2m-1)!!}.
\end{equation}
For $m=1$, one has 
\begin{equation*}
T'_n(\pm 1)=(\pm1)^{n+1} \cdot n^2,
\end{equation*}
which means 
\begin{equation}\label{eq:pm12}
T'_n(\pm 1)=n^2
\end{equation}
as $n$ is odd from (\ref{PermuteCondition}).
From $T'_n(\pm 1) \equiv 1\pmod{p^w}$, it follows that $p^w \mid (n^2-1)$.
Referring to \eqref{eq:t'n1}, one obtains 
$v_p\left(\frac{T^{(m)}_n(\pm 1) \cdot p^m}{m!}\right)
=v_p\left(\prod\limits_{j=2}^{m-1}(n^2-j^2)\right)+v_p(n^2)+v_p(n^2-1)-v_p((2m-1)!!)+m-v_p(m!)$,
implying
\begin{equation}\label{eq:vp:t'npm}
v_p\left(\frac{T^{(m)}_n(\pm 1) \cdot p^m}{m!}\right)\geq 
w-v_p((2m-1)!!)+m-v_p(m!).
\end{equation} 
From $(2m-1)!!=\frac{m\cdot(2m-1)!}{2^{m-1} \cdot m!}$, one has 
\begin{equation}\label{eq:2m-1!!}
v_p(m!)+v_p((2m-1)!!)=v_p(m)+v_p((2m-1)!).
\end{equation}
According to Legendre's formula,
$v_p(x!)=\frac{x-s_p(x)}{p-1}$, 
one has $v_p((2m-1)!)=\frac{2m-1-s_p(2m-1)}{p-1}$.
As $s_p(x)\ge 1$ for $x>0$ and $p\geq 5$,
it further deduces
$v_p((2m-1)!)\leq \frac{2m-2}{p-1}\leq \frac{m-1}{2}$. 
When $m \geq 5$, one has $m^2\leq 5^{m-3}$, namely 
$\log_5(m)\leq \frac{m-3}{2}$.
From $v_p(x)\leq\log_p(x)\leq\log_5(x)$ for $x\ge 1$, one has
$v_p(m)+v_p((2m-1)!)\leq\frac{m-3}{2}+\frac{m-1}{2}=m-2$.
Then, from \eqref{eq:vp:t'npm} and \eqref{eq:2m-1!!}, one has 
\(
v_p\left(\frac{T^{(m)}_n(\pm 1) \cdot p^m}{m!}\right) \geq
w+m-(m-2)=w+2
\)
for any $m \geq 5$.
When $m \in \{2, 3, 4\}$, one can verify 
the preceding inequality
by a direct calculation using \eqref{eq:vp:t'npm}.
Therefore, $\frac{T^{(m)}_n(\pm 1) \cdot p^m}{m!} \equiv 0\pmod{p^{w+2}}$ holds for $m \geq 2$.
\end{proof}

\begin{lemma}\label{lemma:tn(-1+p):equal}
The congruence $T'_n(p-1) \equiv 1\pmod{p^w}$ holds if and only if $T'_n(-1) \equiv 1\pmod{p^w}$, where $w \in \mathbb{N}^+$.
\end{lemma}
\begin{proof}
As $T_n(x)$ is a polynomial of degree $n$, by Taylor's formula, one has
\begin{equation}\label{eq:t'n:(-1+p)}
 T'_n(-1+p)=T'_n(-1)+\sum^n_{m=2}\frac{ T^{(m)}_n(-1)\cdot p^{m-1}}{(m-1)!}.
\end{equation}
If $T'_n(-1)\equiv 1\pmod{p^w}$, $\frac{T^{(m)}_n(-1)\cdot p^m}{m!}$ 
is a multiple of $p^{w+2}$ from Lemma~\ref{lemma:pw:Tn'1}.
Then, $\frac{T^{(m)}_n(-1)\cdot p^{m-1}}{(m-1)!}$
is a multiple of $m\cdot p^{w+1}$ when $m\ge 2$.
Thus, $T'_n(p-1) \equiv 1\pmod{p^w}$, and the sufficient part of this Proposition is proved.

If $T'_n(p-1) \equiv 1\pmod{p^w}$, one can prove
\begin{equation}\label{eq:tn:temp}
T'_n(-1) \equiv 1\pmod{p^t}
\end{equation}
by mathematical induction on $t$.
When $t=1$, it follows from \eqref{eq:t'n:(-1+p)} that
$T'_n(p-1) \equiv T'_n(-1)\pmod{p}$, which means congruence \eqref{eq:tn:temp} holds for $t=1$.
Assume that congruence \eqref{eq:tn:temp} holds for $t=e\leq w-1$, i.e., $T'_n(-1) \equiv 1\pmod{p^e}$.
Then, by Lemma~\ref{lemma:pw:Tn'1}, one can get
$\frac{T^{(m)}_n(-1) \cdot p^{m}}{m!}$ is a multiple of $p^{e+2}$, yielding $\frac{T^{(m)}_n(-1) \cdot p^{m-1}}{(m-1)!} \equiv 0\pmod{p^{e+1}}$ when $m \geq 2$.
When $t=e+1$, using this congruence and \eqref{eq:t'n:(-1+p)}, one has $T'_n(p-1) \equiv T'_n(-1)\pmod{p^{e+1}}$.
Thus, congruence \eqref{eq:tn:temp} also holds for $t=e+1$.
This completes the proof of congruence \eqref{eq:tn:temp}, yielding $T'_n(-1) \equiv 1\pmod{p^w}$.
Consequently, the necessary part of this lemma is proved.
\end{proof}

\begin{lemma}\label{lemma:dTn:x}
When $m \geq 2$, the $m$-order derivative of the Chebyshev polynomial $T_n(x)$, $T^{(m)}_n(x)$, satisfies the recurrence relation
\begin{multline}\label{eq:dTnx:i}
   (x^2-1)T^{(m)}_n(x)=
  \left(n^2-(m-2)^2 \right)T^{(m-2)}_n(x)\\-(2m-3)xT^{(m-1)}_n(x),
\end{multline}
where $T^{(0)}_n(x)=T_n(x)$ and $T^{(1)}_n(x)=T'_n(x)$ represents the first derivative of $T_n(x)$.
\end{lemma}
\begin{proof}
The lemma is proved by mathematical induction on $m$.
Considering the relation between the two kinds of Chebyshev polynomials,
\begin{equation}\label{eq:relation}
T'_n(x)=nU_{n-1}(x),
\end{equation}
and
$U'_{n-1}(x)=\frac{nT_n(x)-xU_{n-1}(x)}{x^2-1}$,
one has
\(
T^{(2)}_n(x)=\frac{n^2T_n(x)-xT'_n(x)}{x^2-1}.
\)
Thus, \eqref{eq:dTnx:i} holds for $m=2$.
Assume that \eqref{eq:dTnx:i} holds for $m=e>2$.
When $m=e+1$, applying the quotient rule for differentiation, one obtains
$T^{(e+1)}_n(x)=\frac{dT^{(e)}_n(x)}{dx}=\frac{(n^2-(e-1)^2)T^{(e-1)}_n(x)}{x^2-1}-\frac{(2e-1)xT^{(e)}_n(x)}{x^2-1}$.
Hence, \eqref{eq:dTnx:i} holds for $m=e+1$.
This completes the proof by induction.
\end{proof}

\begin{lemma}\label{lemma:T'n(x):n}
Given an integer $s$ satisfying $s \not\equiv \pm 1\pmod{p}$,
if $T_n(s) \equiv s\pmod{p^w}$, then one has
$T'_n(s) \equiv \pm n\pmod{p^w}$,
where $w\in\mathbb{N}^+$.
\end{lemma}
\begin{proof}
From \cite[Theorem 2.20]{Lidl:Dick:1993} with $a=1$, one has
a Pell equation
\begin{equation}\label{eq:Tn2:unx}
T^2_n(x)-(x^2-1) U^2_{n-1}(x)=1.
\end{equation}
Since $T_n(s) \equiv s\pmod{p^w}$ and $p\nmid (s^2-1)$, it follows that $U^2_{n-1}(s) \equiv 1\pmod{p^w}$, which means that the order of $U_{n-1}(s)$ in the multiplicative group of the ring $\Zp{w}$ is 
$1$ or $2$. 
Hence,
\begin{equation}\label{eq:U}
U_{n-1}(s) \equiv \pm 1\pmod{p^w}.
\end{equation}
It follows from \eqref{eq:relation} that $T'_n(s) \equiv \pm n\pmod{p^w}$.
\end{proof}

\begin{lemma}\label{lemma:Tn:0pw}
Given $s \in \mathcal{Z}_p$, if \( T_n(s) \equiv s\pmod{p^w} \) and \(T'_n(s) \equiv 1\pmod{p^w} \), then
\(T'_n(s+hp) \equiv T'_n(s)\pmod{p^{w+1}} \)
and
\(\frac{T^{(m)}_n(s+hp) \cdot p^m}{m!} \equiv 0\pmod{p^{w+2}} \)
for \( m\ge 2 \), where
$w\in\mathbb{N}^+$ and $h\in\mathbb{Z}$.
\end{lemma}
\begin{proof}
Since \( T_n(x) \) is a polynomial of degree \( n \), by Taylor's formula one has
\begin{equation}\label{eq:tn(m)}
 T^{(m)}_n(s+x\cdot p)=T^{(m)}_n(s)+\sum^n_{i=m+1} \frac{(x\cdot p)^{i-m}}{(i-m)!} T^{(i)}_n(s)
\end{equation}
for any non-negative integer $m$ and $x \in \mathbb{Z}$.
Depending on the value of $s$, the proof of the lemma is divided into the following three cases:
\begin{itemize}[\newItemizeWidth]
\item $s=1$: Multiplying both sides of \eqref{eq:tn(m)} by \( \frac{p^m}{m!} \), one obtains
\begin{multline}\label{multi2sides}
\frac{T^{(m)}_n(s+x \cdot p) \cdot p^m}{m!}=
\frac{T^{(m)}_n(s) \cdot p^m}{m!}+\\
\sum^n_{i=m+1} x^{i-m} \binom{i}{m} \frac{T^{(i)}_n(s) \cdot  p^i}{i!}
\end{multline}
for any $m\in \mathbb{N}^+$.
By Lemma~\ref{lemma:pw:Tn'1}, one knows that
$\frac{T^{(i)}_n(1) \cdot  p^i}{i!} \equiv 0\pmod{p^{w+2}}$
for $i\ge 2$.
Substituting \( (s, x)=(1, h) \) into~\eqref{multi2sides}, one concludes that 
\(\frac{T^{(m)}_n(1+hp) \cdot p^m}{m!} \equiv \frac{T^{(m)}_n(1) \cdot p^m}{m!}\pmod{p^{w+2}} \) for $m\ge 1 $. 
It implies
\(T'_n(1+hp) \cdot p \equiv T'_n(1) \cdot p\pmod{p^{w+2}}\)
and \(\frac{T^{(m)}_n(1+hp) \cdot p^m}{m!} \equiv 0\pmod{p^{w+2}} \)
for \(m\ge 2\). The former is equivalent to \( T'_n(1+hp) \equiv T'_n(1)\pmod{p^{w+1}}\).

\item 
 \( s=p-1 \): Referring to Lemma~\ref{lemma:tn(-1+p):equal}, one has
\(T'_n(-1) \equiv 1\pmod{p^w} \) from \( T'_n(p-1) \equiv 1\pmod{p^w} \).
By Lemma~\ref{lemma:pw:Tn'1}, it follows that
\( \frac{T^{(i)}_n(-1) \cdot p^i}{i!}\equiv 0\pmod{p^{w+2}} \) for $i\ge 2$.
Substituting \( (s, x)=(-1, h+1) \) into~\eqref{multi2sides}, one obtains
\( \frac{T^{(m)}_n(p-1+hp) \cdot p^m}{m!} \equiv \frac{T^{(m)}_n(-1) \cdot p^m}{m!}\pmod{p^{w+2}}\)
for \(m\ge 1\).
This implies
\(
T'_n(p-1+hp) \cdot p\equiv T'_n(-1)\cdot p\pmod {p^{w+2}},
\) 
and
\( \frac{T^{(m)}_n(p-1+hp) \cdot p^m}{m!} \equiv 0\pmod{p^{w+2}}\)
for \(m\ge 2\). 
The former is equivalent to \( T'_n(p-1+hp) \equiv T'_n(-1)\pmod{p^{w+1}} \). 
Since this congruence holds for any $h\in\mathbb{Z}$,
one has $T'_n(p-1+hp) \equiv T'_n(p-1)\pmod{p^{w+1}}$.

\item \(s \notin \{1, p-1\} \): Referring to Lemma~\ref{lemma:T'n(x):n}, one has \( T'_n(s) \equiv \pm n\pmod{p^w} \), which implies \( n^2 \equiv 1\pmod{p^w} \) due to $T'_n(s) \equiv 1\pmod{p^w}$. Consequently, one has \( n^2 T_n(s)-s T'_n(s) \equiv T_n(s)-s T'_n(s) 
\equiv s-s\cdot 1\equiv 0\pmod{p^w} \). 
From Lemma~\ref{lemma:dTn:x} with $(m, x)=(2, s)$, it follows that
$(s^2-1) T^{(2)}_n(s)=n^2 T_n(s)-s T'_n(s)$, and therefore,
$(s^2-1) T^{(2)}_n(s)\equiv 0\pmod{p^w}$.
Since \( p\nmid (s^2-1)\) in this case, one can deduce that \( T^{(2)}_n(s) \equiv 0\pmod{p^w}\). 
Applying the same reasoning to Lemma~\ref{lemma:dTn:x} with $(m, x)=(3, s)$, one obtains $T^{(3)}_n(s)\equiv 0\pmod{p^w}$.
By induction on \(m\) in \eqref{eq:dTnx:i}, it follows that
\[
T^{(m)}_n(s) \equiv 0\pmod{p^w}
\]
for \(m\geq 2\).
Moreover, since \(v_p(x!)\leq \frac{x-1}{p-1}\), one finds that
$v_p\left(\frac{(hp)^{i-m}}{(i-m)!} T^{(i)}_n(s)\right)\ge i-m-\frac{(i-m)-1}{p-1}+w=\frac{(i-m)(p-2)+1}{p-1}+w$ for $i\ge 2$.
Thus, one concludes
\(
v_p\left(\frac{(hp)^{i-m}}{(i-m)!} T^{(i)}_n(s)\right)\ge 1+w 
\)
for $i\ge m+1$. 
Substituting $x=h$ into~\eqref{eq:tn(m)}, one can conclude that 
\(
T'_n(s+hp) \equiv T'_n(s)\pmod{p^{w+1}},
\)
and
\[
T^{(m)}_n(s+hp) \equiv 0\pmod{p^w}
\]
for \(m\ge 2\).
Finally, noting that $p\ge 5$, and $v_p(m!)=\frac{m-s_p(m)}{p-1}
\leq \frac{m-1}{4}$, one deduces that
$v_p\left(T^{(m)}_n(s+hp)\frac{p^m}{m!}\right)
\ge w+m-\frac{m-1}{4}=
w+\frac{3m+1}{4}$, which implies
$v_p\left(T^{(m)}_n(s+hp)\frac{p^m}{m!}\right)\geq w+2$
for $m\ge 2$ as $v_p(x)\in\mathbb{Z}$.
\end{itemize}
\end{proof}

\begin{proposition}\label{lemma:Tn(x):w}
Given $s\in\mathcal{Z}_p$, if $T_n(s)\equiv s\pmod p$ and $T'_n(s)\equiv 1\pmod{p^w}$, then one has
\begin{equation}\label{eq:tn:pt1}
T_n(s+hp) \equiv s+hp\pmod{p^{w+1}}, 
\end{equation}
where $h\in\mathbb{Z}$ and $w\in\mathbb{N}^+$.
\end{proposition}
\begin{proof}
Referring to \eqref{eq:tn(m)} with $(m, x)=(0, j)$, one obtains
$T_n(s+jp)=T_n(s)+j pT'_n(s)+\sum^n_{i=2}\frac{j^i\cdot p^i}{i!}T^{(i)}_n(s)$,
where $j\in\mathbb{Z}$.
Now, assuming that $T_n(s)\equiv s\pmod{p^{w+1}}$, 
it follows directly that $T_n(s)\equiv s\pmod{p^w}$.
Combining this with the second condition stated in this proposition, one can get
the third term in the above expression of $T_n(s+jp)$ is a multiple of 
$p^{w+2}$ from Lemma~\ref{lemma:Tn:0pw} with $h=0$.
Furthermore, since $T'_n(s)\equiv 1\pmod{p^w}$, one deduces that 
$j pT'_n(s)\equiv j p\pmod{p^{w+1}}$.
Thus, it follows that (\ref{eq:tn:pt1}) holds
for any $h, w$ within the prescribed range. Consequently,
it is sufficient to consider the assumed condition, i.e., the case $h=0$, for this proposition.

Depending on the value of $s$, the proof of congruence~\eqref{eq:tn:pt1} with $h=0$ is divided into the following three cases:
\begin{itemize}[\newItemizeWidth]
    \item $s\in\{0, 1\}$: One has $T_n(s)=s$ from \eqref{eq:recursive} as $n$ is odd.
    Thus, $T_n(s) \equiv s\pmod{p^{w+1}}$.
    
    \item $s=p-1$: One has $T'_n(-1)\equiv 1\pmod{p^w}$ from Lemma~\ref{lemma:tn(-1+p):equal}. Referring to \eqref{eq:tn(m)} with $(m, x, s)=(0, 1, -1)$ and Lemma~\ref{lemma:pw:Tn'1}, one can get
    $T_n(-1+p)=T_n(-1)+pT'_n(-1)+\sum^n_{i=2}\frac{p^i}{i!}T^{(i)}_n(-1)\equiv T_n(-1)+ p\pmod{p^{w+1}}$.
    Since $T_n(1)=1$ and $T_n(x)$ is an odd function, one has $T_n(-1)=-1$. Therefore, $T_n(p-1) \equiv p-1\pmod{p^{w+1}}$.
    
    \item $s\notin\{0, 1, p-1\}$: 
    From $T_n(s)\equiv s\pmod p$, it follows that
    \begin{equation}\label{eq:unep}
    U_{n-1}(s)\equiv \pm 1\pmod p
    \end{equation}
    by \eqref{eq:U} and its precondition with $w=1$.
    As for either case of \eqref{eq:unep}, one can prove that
    \begin{equation}\label{eq:tn:pt}
    T_n(s) \equiv s\pmod{p^t}
    \end{equation}
    holds for $t\leq w+1$ by mathematical induction on $t$.
    When $t=1$, congruence~\eqref{eq:tn:pt} is given by the proposition's condition.
    Assume congruence~\eqref{eq:tn:pt} holds for $t=e<w+1$, i.e., $T_n(s)\equiv s\pmod{p^e}$.
    Then, 
    \begin{equation}\label{eq:unpe}
    U_{n-1}(s)\equiv \pm 1\pmod{p^e}
    \end{equation}
  by \eqref{eq:U}, similar to deriving \eqref{eq:unep}.
  The remaining part of this proof is to prove congruence~\eqref{eq:tn:pt} holds when $t=e+1$ to complete the induction.
\end{itemize}

If the positive case of \eqref{eq:unep} holds, one multiplies both sides of \eqref{eq:s:sunx} with $x=s$ by $(s-\sqrt{s^2-1}) $, yielding the relation
\begin{equation}\label{eq:ss2-1n-1}
\left(s+\sqrt{s^2-1}\right)^{n-1}=c_1+c_2\sqrt{s^2-1},
\end{equation}
where 
\begin{equation}\label{eq:c1c2}
\begin{cases}
c_1=sT_n(s)-(s^2-1)U_{n-1}(s);\\
c_2=sU_{n-1}(s)-T_n(s).
\end{cases}
\end{equation}
Since $U_{n-1}(s)\equiv 1\pmod p$, one has $U_{n-1}(s)\equiv 1\pmod{p^e}$
from \eqref{eq:unpe}.
Therefore, one deduces that $c_1\equiv 1\pmod{p^e}$ and $c_2\equiv 0\pmod{p^e}$.
Furthermore, from the second condition
of this proposition, one obtains $T_n'(s)\equiv 1\pmod{p^e}$ as $e\le w$. This implies
$n\equiv 1\pmod{p^e}$ 
by \eqref{eq:relation}, leading to $n-1=b\cdot p^e$.
Let $(s+\sqrt{s^2-1})^b=d_1+d_2\sqrt{s^2-1}$.
Then, from \eqref{eq:ss2-1n-1}, one has
$
\left(s+\sqrt{s^2-1}\right)^{n-1}=\left(d_1+d_2\sqrt{s^2-1}\right)^{p^e}
$,
yielding
\[
\begin{cases}
c_1=\sum\limits_{i=1}^{p^*}\binom{p^e}{2i}d_2^{2i} 
d_1^{p^e-2i}(s^2-1)^i+d_1^{p^e};\\
c_2=\sum\limits_{i=1}^{p^*}\binom{p^e}{2i-1} d_2^{2i-1} d_1^{p^e-2i+1}(s^2-1)^{i-1} +d_2^{p^e} (s^2-1)^{p^*},
\end{cases}
\]
where $p^*=\frac{p^e-1}{2}$.

Applying Kummer’s theorem~\cite{Mihet:2010:Kummer}, one has $v_p\left(\tbinom{p^e}{j}\right)=e-v_p(j)$, which
is greater than or equal to one for $1\le j\leq p^e-1$.
This implies that $p \mid \binom{p^e}{2i-1}$ in the first expression term of $c_2$,
implying $c_2\equiv d_2^{p^e}\cdot (s^2-1)^{p^*}\pmod{p}$. 
Since $c_2\equiv 0\pmod{p}$ and $p\nmid (s^2-1)$, one has $p\mid d_2$.
Moreover, since $v_p(x)\leq\log_p(x)\leq x-1$ for $x\ge 1$ and 
$\ln(p)\ge 1$, it follows that $v_p\left(\tbinom{p^e}{j}\cdot p^j\right)=e-v_p(j)+j\ge e+1$ for any $j$.
Thus, one obtains 
\begin{equation}\label{eq:c2}
\begin{cases}
c_1\equiv d_1^{p^e}\pmod{p^{e+1}};\\
c_2\equiv 0\pmod{p^{e+1}},
\end{cases}
\end{equation}
as $p^e>e+1$.
Given that $c_1\equiv 1\pmod{p^e}$, one concludes that $d_1^{p^e}\equiv 1\pmod{p^e}$, indicating that the order of $d_1$ in $\mathbb{Z}^*_{p^e}$, the multiplicative group of the ring $\Zp{e}$, is a power of $p$.
Referring to \cite[Corollary 14.12]{Wan:Lecture:2003}, 
the order of any element in $\mathbb{Z}^*_{p^e}$ takes the form $k_1\cdot k_2$, where $k_1\mid (p-1)$ and $k_2\mid p^{e-1}$. 
Hence, one has $k_1=1$, and the order of \(d_1\) in $\mathbb{Z}^*_{p^e}$ divides \(p^{e-1}\), implying that $d_1^{p^{e-1}}\equiv 1\pmod{p^e}$, i.e.,
$d_1^{p^{e-1}}=1+c\cdot p^e$.
Thus, one derives
\begin{equation*}
c_1\equiv (1+c\cdot p^e)^{p} \equiv 1\pmod{p^{e+1}}
\end{equation*}
from \eqref{eq:c2}.
Finally, by substituting the expressions for $c_1$ and $c_2$ from \eqref{eq:c1c2} into the preceding congruence and the second congruence in \eqref{eq:c2}, respectively, one obtains 
\[
T_n(s) \equiv s\pmod{p^{e+1}}.
\]
In the negative case of \eqref{eq:unep}, one multiplies both sides of \eqref{eq:s:sunx} with $x=s$ by $(s+\sqrt{s^2-1})$ 
instead. Then, \eqref{eq:ss2-1n-1} becomes
$(s+\sqrt{s^2-1})^{n+1}=c'_1+c'_2\sqrt{s^2-1}$, 
where 
$c'_1=sT_n(s)+(s^2-1)U_{n-1}(s)$ and $c'_2=sU_{n-1}(s)+T_n(s)$ are both integers.
The proof for this case follows the same reasoning as in the positive case and is therefore omitted here.
\end{proof}

\subsection{Explicit Expression of the Least Period of the Chebyshev Integer Sequence}\label{subsec:cheby:integer}

Given $s\in\mathcal{Z}_{p^k}$, one can know a Chebyshev integer sequence $\{T_n^i(s)\bmod p^k\}_{i\ge 0}$ is periodic from the permutation property of the Chebyshev polynomial over the ring $\mathbb{Z}_{p^k}$. 
Under the assumption that there exists a positive integer $v\ge 2$
such that
\begin{equation}\label{vthreshold1}
\begin{cases}
T_n^{N_s\cdot l_s}(s)\equiv s\pmod{p^v},\\
T_n^{N_s\cdot l_s}(s)\not\equiv s\pmod{p^{v+1}}, 
\end{cases}
\end{equation}
the least period of the Chebyshev integer sequence was discussed in~\cite[Theorem 2]{Yoshioka:ChebyshevPk:TCAS2:2018}, 
where $N_x$ is the least period of the sequence $\{T_n^i(x)\bmod p\}_{i\ge 0}$ and
\begin{equation}\label{eq:lx:error}
l_x=
\begin{cases}
p             & \mbox{if } T'_{n^{N_x}}(x) \equiv 1\pmod p; \\
\ord(T'_{n^{N_x}}(x)) & \mbox{if } T'_{n^{N_x}}(x) \not\equiv 1\pmod p.
\end{cases}
\end{equation}
We hypothesize that the purpose of setting $l_s=p$
in \eqref{eq:lx:error} when $T'_{n^{N_s}}(s) \equiv 1\pmod{p}$ is to ensure $v\ge 2$ in condition~(\ref{vthreshold1}).
In fact, when $T'_{n^{N_s}}(s) \equiv 1\pmod p$,
it follows that \(\ord(T'_{n^{N_s}}(s))=1\); \(T_{n^{N_s}}(s) \equiv s\pmod{p^2}\) from
the definition of $N_x$ and Proposition~\ref{lemma:Tn(x):w}, which means \(v \geq 2\) in condition~(\ref{vthreshold1}) with \(l_s=1\). Thus, we correct \eqref{eq:lx:error} to
\begin{equation*}
l_x=\ord(T'_{n^{N_x}}(x))
\end{equation*}
regardless of the input parameter of \(\ord(\cdot)\). 

When \(s \not\equiv \pm 1\pmod{p}\), since \(T_{n^{N_s}}(s) \equiv s\pmod{p}\), it follows from Lemma~\ref{lemma:T'n(x):n} with $w=1$ that \(T'_{n^{N_s}}(s) \equiv \pm n^{N_s}\pmod{p}\).
When $s\equiv1\pmod{p}$, one has $T_n(s)\equiv 1\pmod p$ by induction on the degree of the Chebyshev polynomial from zero. 
When $s\equiv-1\pmod{p}$, one has $T_n(s)\equiv -1\pmod p$ in the same way as $n$ is odd.
In either case, one has $T_n(s)\equiv s\pmod p$, implying $N_s=1$.
It further deduces $T'_{n^{N_s}}(s)\equiv n^2\pmod{p}$ from \eqref{eq:pm12}.
Since \(\gcd(n, p)=1\), it follows that
\(p \nmid T'_{n^{N_s}}(s)\), which
justifies the definition of \(l_x\).

In general, the least period of the sequence \(\{T_n^i(s) \mod p^k\}_{i \geq 0}\) increases monotonically
with respect to \(k\) for any non-trivial initial state.
Furthermore, Lemma~\ref{lemma:period:change} establishes that this period strictly increases once \(k\) surpasses a certain threshold.
Lemma~\ref{lemma:TniN} provides the explicit form of any iteration of \(T_n^{N\cdot h}(x)\) over $\Zp{w+1}$ employing the geometric series of its derivative.
Lemma~\ref{lemma:T'l} addresses the commutation between the
derivative and the exponent of \(T_n^N(x)\) with respect to modulo \(p\).
Lemma~\ref{lemma:xmodp2} describes the congruence relation between a geometric series on certain specific integers and \(p^2\).
Finally, with the aid of these four lemmas and Proposition~\ref{lemma:Tn(x):w}, 
one can establish a rigorous statement regarding the least period of the Chebyshev integer sequence iterated from any non-trivial initial state in \(\mathcal{Z}_{p^k}\) in Theorem~\ref{Theorem:periodValue},
where the three key parameters $N_s$, $l_s$, and $v$ are determined by the given conditions in turn.

\begin{lemma}\label{lemma:period:change}
Given $s\in\mathcal{Z}_{p^k}\setminus\{0, 1\}$ and $w\in\mathbb{N}^+$,
the least period of the sequence $\{T_n^i(s)\bmod p^k\}_{i\ge 0}$
is greater than \(N\) when
$k>\max\{e \vvert T_n^N(s)\equiv s\pmod{p^e}\}$,
where $N$ is the least period of the sequence $\{T_n^i(s)\bmod p^w\}_{i\ge 0}$.
\end{lemma}
\begin{proof}
Let $g(x)=T_{n^N}(x)-x$. From \eqref{eq:relation}, one has 
$g'(x)=n^N U_{n^N-1}(x)-1$.
Since $U_{n^N-1}(x)=\sum_{k=0}^{\lfloor (n^N-1)/2 \rfloor}\binom{n^N}{2k+1}(x^2-1)^k x^{n^N-1-2k}$ as stated in \cite[Theorem 1]{2019NotesOE},
it yields that $g(x)>0$ for $x>1$, given that $g(1)=0$, and $g'(x)>0$ when $x>1$.
Thus, one concludes that $T^N_n(s)=T_{n^N}(s)>s$ for any $s\in\mathcal{Z}_{p^k}\setminus\{0, 1\}$, meaning that $w_s=\max\{e\vvert T_n^N(s)\equiv s\pmod{p^e}\}$ is finite.
By the definition of \(N\), one has \(T_n^N(s) \equiv s\pmod{p^w}\), i.e.,
\(T_n^N(s)=s + c \cdot p^w\).
Furthermore, it follows that \(w_s \geq w\) and \(c=c_p \cdot p^{w_s-w}\)
by the definition of \(w_s\), 
where $\gcd(c_p, p)=1$.

Let $N^*$ denote the least period of the sequence $\{T_n^i(s)\bmod p^k\}_{i\ge 0}$, where $k>w_s$.
Assume that $N^*<{N}$, one has $T_n^{N^*}(s)\equiv s\pmod{p^w}$
since $k>w$, implying that
$N^*$ is the period of the sequence $\{T_n^i(s)\bmod p^w\}_{i\ge 0}$.
However, this contradicts the definition of \(N\). Therefore, \(N^* \geq N\).
Assuming \(N^*=N\), one would have \(T_n^N(s) \equiv s\pmod{p^k}\).
Moreover, since \(T_n^N(s)=s+c_p \cdot p^{w_s}\) and \(k > w_s\), it follows that
\(p \mid c_p\), which contradicts the fact that \(\gcd(c_p, p)=1\).
Hence, one concludes that $N^* > N$, thereby proving the lemma.
\end{proof}

\begin{lemma}\label{lemma:TniN}
Given two positive integers $s$ and $w$, if \begin{equation}\label{assumec}
 \begin{cases}
 T^N_n(s)\equiv s\pmod{p^w}, \\
 T^N_n(s)\not\equiv s\pmod{p^{w+1}}, 
 \end{cases}
\end{equation}
then
\begin{equation}\label{ch:eq:Tinx}
T^{N \cdot h}_n(s)\equiv s+b\cdot p^w \sum_{j=0}^{h-1}(T'_{n^N}(s))^j
\pmod{p^{w+1}}, 
\end{equation}
where $h, N, w\in\mathbb{N}^+$, $b\in\mathbb{Z}$, and $b\nmid p$.
Specifically, the base and the residue in the preceding congruence~\eqref{ch:eq:Tinx} remain congruent modulo $p^{w+2}$ if $w\geq 2$.
\end{lemma}
\begin{proof}
The congruence~(\ref{ch:eq:Tinx}) is proved via mathematical induction on $h$.
When $h=1$, one has
$T^N_n(s)\equiv s+b\cdot p^w\pmod{p^{w+1}}$
from condition~(\ref{assumec}).
Thus, congruence~(\ref{ch:eq:Tinx}) holds for $h=1$.
Suppose that congruence~(\ref{ch:eq:Tinx}) holds for $h=e>1$, namely,
$T^{N \cdot e}_n(s)\equiv s+B\pmod{p^{w+1}}$, where
$B=b\cdot p^w \sum_{j=0}^{e-1}(T'_{n^N}(s))^j$.
When $h=e+1$, one computes
$T^{{N \cdot (e+1)}}_n(s)=T^N_n(T^{N \cdot e}_n(s))\equiv T^N_n(s+B)\pmod{p^{w+1}}$.
Applying Taylor's formula, one obtains
\begin{equation}\label{eq:TnNx+b}
T^N_n(s+B)=T^N_n(s)+\sum^{n^N}_{i=1}B^{i}\frac{T^{(i)}_{n^N}(s)}{i!}.
\end{equation}
Since $p^w\mid B$, $p^{w+1}\mid B^i$ for any $i\ge 2$. Moreover, $\frac{T^{(i)}_{n^N}(s)}{i!}$ is an integer.
Hence,
\begin{equation*}
\begin{aligned}
T^{N\cdot (e+1)}_n(s) & \equiv 
s+b\cdot p^w+B\cdot T'_{n^N}(s)\pmod{p^{w+1}} \\
& \equiv s+b\cdot p^w \sum_{j=0}^e(T'_{n^N}(s))^j\pmod{p^{w+1}}.
\end{aligned} 
\end{equation*}
Therefore, congruence~(\ref{ch:eq:Tinx}) also holds for $h=e+1$,
completing the induction and thus proving congruence~(\ref{ch:eq:Tinx}).

From condition~(\ref{assumec}), one has $T^N_n(s)\equiv s+b'\cdot p^w\pmod{p^{w+2}}$, where $b'\nmid p$.
If $w\ge 2$, one has $iw\ge w+2$ for any $i\ge 2$.
Therefore, $p^{w+2}$ can divide $B^i$ in \eqref{eq:TnNx+b} also for any $i\ge 2$.
Applying the same inductive framework as above, one can prove that the base and the residue in congruence~\eqref{ch:eq:Tinx}, with $b=b'$ remain congruent modulo $p^{w+2}$ if $w\geq 2$.
\end{proof}

\begin{lemma}\label{lemma:T'l}
Given a non-negative integer $s$, if $T^N_n(s)\equiv s\pmod{p}$,
then $(T^{N\cdot h}_n(s))'\equiv (T'_{n^N}(s))^h\pmod p$ holds for any positive integer $h$.
\end{lemma}
\begin{proof}
The lemma can be obtained by replacing $n$ in~\cite[Lemma 6]{Yoshioka:ChebyshevPk:TCAS2:2018} with $n^N$.
\end{proof}

\begin{lemma}
\label{lemma:xmodp2}
For any non-negative integer $s$ satisfying $s\equiv 1\pmod p$, one has
\begin{equation}\label{eq:xmodp2}
\sum_{i=0}^{p-1}s^i\equiv p\pmod{p^2}
\end{equation}
and
$p=\min\{h \vvert \sum_{i=0}^{h-1}s^i\equiv 0\pmod p, h\in \mathbb{N}^+\}$.
\end{lemma}
\begin{proof}
Since $s\equiv 1\pmod p$, one has $s=c\cdot p+1$ and
$s^i=(c\cdot p+1)^i$.
Then, $s^i = (c\cdot p+1)^i \equiv (1+i\cdot c\cdot p) \pmod {p^2}$. 
It means $\sum_{i=0}^{j-1}s^i \equiv \sum_{i=0}^{j-1}(1+i\cdot c\cdot p) \pmod {p^2}$, 
yielding $\sum_{i=0}^{j-1}s^i \equiv j+\frac{j\cdot (j-1)}{2}\cdot c\cdot p\pmod{p^2}$.
Setting $j=p$, one can obtain congruence~\eqref{eq:xmodp2}.
Furthermore, from the preceding congruence, one can verify that $p$ is the least positive integer $j$ such that
$\sum_{i=0}^{j-1}s^i\equiv 0\pmod p$.
\end{proof}

\begin{theorem}\label{Theorem:periodValue}
Given $s\in\mathcal{Z}_{p^k}\backslash \{0, 1\}$, the least period of sequence $\{T_n^i(s)\bmod p^k\}_{i\ge 0}$ is $L(n, p, k, s)=N_s\cdot l_s\cdot p^{k-v}$ when $k\ge v$, where $N_s$ is the least period of sequence $\{T_n^i(s)\bmod p\}_{i\ge 0}$, 
$l_s$ is the least positive number $j$ such that $(T'_{n^{N_s}}(s))^j\equiv 1\pmod p$,
and
$v=\max\{e \vvert T_n^{N_s\cdot l_s}(s)\equiv s\pmod{p^e}\}.$
\end{theorem}
\begin{proof}
From the definition of $N_s$, one has $T^{N_s}_n(s)\equiv s\pmod{p}$. Referring to Lemma~\ref{lemma:period:change}, 
there exists a positive integer $w_s$ such that
\begin{equation*}
 \begin{cases}
 T^{N_s}_n(s)\equiv s\pmod{p^{w_s}};\\
 T^{N_s}_n(s)\not\equiv s\pmod{p^{{w_s}+1}},
 \end{cases}
\end{equation*}
namely $w_s=\max\{e \vvert T_n^{N_s}(s)\equiv s\pmod{p^e}\}$,
for a given $N_s$ according to any $s$.
Due to the minimality of $N_s$, one has $L(n, p, k, s)=N_s$ when $k\leq w_s$.

If $l_s\neq 1$, setting $(N, w, h)=(N_s, w_s, l_s)$ in Lemma~\ref{lemma:TniN}, one obtains 
\begin{equation}\label{eq:TNss}
T^{N_s\cdot l_s}_n(s)\equiv s+b\cdot p^{w_s} \sum_{j=0}^{l_s-1}(T'_{n^{N_s}}(s))^j\pmod{p^{w_s+1}}.
\end{equation}
By definition, $l_s$ is the least integer $j$ such that
\[
(T'_{n^{N_s}}(s))^{j}-1=(T'_{n^{N_s}}(s)-1)\sum_{j=0}^{j-1}(T'_{n^{N_s}}(s))^j \equiv 0\pmod{p},
\]
which means \(\sum_{j=0}^{l_s-1}(T'_{n^{N_s}}(s))^j \equiv 0\pmod{p}\) since $l_s\neq 1$ implies
$(T'_{n^{N_s}}(s)-1)\not\equiv 0\pmod{p}$.
Thus, one deduces
\begin{equation}\label{eq:ws1}
T^{N_s\cdot l_s}_n(s)\equiv s\pmod{p^{w_s+1}}
\end{equation}
from \eqref{eq:TNss}.
Referring to Lemma~\ref{lemma:period:change}, there exists an integer $v$ such that
\begin{equation}\label{vthreshold}
\begin{cases}
T_n^{N_s\cdot l_s}(s)\equiv s\pmod{p^v}; \\
T_n^{N_s\cdot l_s}(s)\not\equiv s\pmod{p^{v+1}}.
\end{cases}
\end{equation}
If $l_s=1$, i.e. $T'_{n^{N_s}}(s)\equiv 1 \pmod{p}$,
from $T_{n^{N_s}}(s)\equiv s\pmod p$ and Proposition~\ref{lemma:Tn(x):w} with $(w, h)=(1, 0)$, one has $T_{n^{N_s}}(s)=T^{N_s}_n(s)\equiv s\pmod{p^2}$,
implying $w_s\ge 2$.
In either case, condition~\eqref{vthreshold} holds,
where
\[
\begin{cases}
v=w_s       & \mbox{if } l_s=1; \\
v\ge w_s+1  & \mbox{if } l_s\neq 1,
\end{cases}
\]
and $v\geq 2$ as $w_s\in\mathbb{N}^+$.
By the minimality of $l_s$, one concludes that $L(n, p, k, s)=N_s \cdot l_s$ for $w_s < k \le v$ if $l_s \neq 1$. If $l_s=1$, the same expression holds for $k \le v$.

When $k\ge v$, one can state the theorem as
\begin{equation}
\min\{i \mid T^i_n(s)\equiv s\pmod{p^{v+t}}, i\in \mathbb{N}^+ \}=N_s\cdot l_s\cdot p^t,
\label{eq:theorem1}
\end{equation}
where $t=k-v$. 
As shown in the previous discussion, there is a tight connection between \eqref{eq:theorem1}
and
\begin{equation}\label{eq:change:pi}
 \begin{cases}
 T^{N_s\cdot l_s\cdot p^t}_n(s)\equiv s\pmod{p^{v+t}}; \\
 T^{N_s\cdot l_s\cdot p^t}_n(s)\not\equiv s\pmod{p^{v+t+1}}.
 \end{cases}
\end{equation}
One now proceeds to prove them together by induction on $t$.
From the analysis of the case $k=v$, one sees that \eqref{eq:theorem1} and 
condition~\eqref{eq:change:pi} hold for $t=0$.
Assume that they both hold for $t=e\in[1, k-v)$, i.e.,
$\min\{i \mid T^i_n(s)\equiv s\pmod{p^{v+e}}, i\in \mathbb{N}^+\}=N_s\cdot l_s\cdot p^e$,
and
\begin{equation}\label{eq:change:pe}
 \begin{cases}
 T^{N_s\cdot l_s\cdot p^e}_n(s)\equiv s\pmod{p^{v+e}};\\
 T^{N_s\cdot l_s\cdot p^e}_n(s)\not\equiv s\pmod{p^{v+e+1}}.
 \end{cases}
\end{equation} 
Setting $(N, w)=(N_s\cdot l_s\cdot p^e, v+e)$ in Lemma~\ref{lemma:TniN}, it follows from \eqref{eq:change:pe} that
\begin{multline}
\label{eq:Tni:Nlx}
T^{N_s\cdot l_s\cdot p^e \cdot h}_n(s)\equiv\\
s+ 
b'\cdot p^{v+e} \sum_{j=0}^{h-1}(T'_{n^{N_s\cdot l_s\cdot p^e}}(s))^j\pmod{p^{v+e+2}},
\end{multline}
where $v+e>2$ and $b'\nmid p$.
One verifies from Lemma~\ref{lemma:T'l} with $(N, h)=(N_s, p^e \cdot l_s)$
that
$T'_{n^{N_s\cdot p^e\cdot l_s}}(s)\equiv (T'_{n^{N_s}}(s))^{p^e\cdot l_s}\pmod p$.
From the definition of $l_s$, it follows that
$T'_{n^{N_s\cdot l_s\cdot p^e}}(s)\equiv 1\pmod p$.
By Lemma~\ref{lemma:xmodp2}, one has
\begin{equation}\label{eq:tn2p}
\sum_{j=0}^{p-1}(T'_{n^{N_s\cdot l_s\cdot p^e}}(s))^j\equiv p\pmod{p^2}
\end{equation}
and
$
p=\min\{h \vvert \sum_{j=0}^{h-1}(T'_{n^{N_s\cdot l_s\cdot p^e}}(s))^j\equiv 0\pmod p, h\in \mathbb{N}^+ \}
$.
Adopting the relation on $T'_{n^{N_s\cdot l_s\cdot p^e}}(s)$ in \eqref{eq:tn2p} and the preceding equation into congruence~\eqref{eq:Tni:Nlx} with $h=p$,
one can conclude that \eqref{eq:theorem1} and condition~\eqref{eq:change:pi} both hold for $t=e+1$.
The induction thus completes their proof.
\end{proof}

The periods of the sequence $\{T_n^i(s)\bmod p^k\}_{i\ge 0}$ with $(n, p)\in\{(29, 7), (19, 13)\}$ and $s\in\{0, 1, \cdots, 5\}$ are presented in Table~\ref{table:period:n29p7}, where 
the values of intermediate $l_s$ are also listed.
When $s\in\{3, 4\}$, one observes via numerical simulation that
$L(29, 7, 3, s)=14$ and $L(29, 7, 4, s)=98$ via numerical simulation.
According to the theory analysis, one can compute $(N_s, l_s, v)=(1, 2, 2)$, which further leads to the conclusion that $L(29, 7, k, s)=2\cdot 7^{k-2}$ for any $k\geq 2$, by applying Theorem~\ref{Theorem:periodValue}.
Similarly, for $s=3$, one finds that $L(19, 13, 5, 3)=156$ by numerical simulation. Alternatively, by calculating $(N_s, l_s, v)=(3, 4, 4)$, one deduces that $L(19, 13, k, 3)=12\cdot 13^{k-4}$ for any $k\geq 4$, again utilizing Theorem~\ref{Theorem:periodValue}. In summary, the data in Table~\ref{table:period:n29p7} align with the predictions of Theorem~\ref{Theorem:periodValue}.

From \eqref{eq:recursive} and \eqref{PermuteCondition}, one has $T_n(s)=s$ when $s \in \{0, 1\}$. Hence, the period of the Chebyshev integer sequence, when starting from $s\in\{0, 1\}$, remains invariant for any pair $(n, p)$.
Drawing on the proof of Theorem~\ref{Theorem:periodValue}, one may observe 
two more key rules governing the periodic behavior of the Chebyshev integer sequence $\{T_n^i(s)\bmod p^k\}_{i\ge 0}$ as the parameter $k$ increases:
\begin{itemize}[\newItemizeWidth]
\item If $l_s\neq 1$, the period of the Chebyshev integer sequence increases by a factor of $l_s$ as $k$ increases, up to a certain threshold depending on $s$. Beyond another threshold, the period expansion ratio stabilizes at $p$. 
The period remains unchanged when $k$ is between the two different thresholds.
The data for $s \in \{3, 4\}$ when $(n, p)=(29, 7)$, and for 
$s\in \{2, 3, 4, 5\}$ when $(n, p)=(19, 13)$, shown in Table~\ref{table:period:n29p7}, follow this rule.

\item If $l_s=1$, the period expansion ratio becomes constant at $p$ when $k$ is greater than a threshold.
The data for $s \in \{2, 5\}$ when $(n, p)=(29, 7)$, as listed in Table~\ref{table:period:n29p7}, exemplify this rule.
\end{itemize}

\setlength\tabcolsep{2pt} 
\begin{table}[!htb]
    \caption{The least period of sequence $\{T_n^i(s)\bmod p^k\}_{i\ge 0}$.}
    \centering 
    \begin{tabular}{*{12}{c|}c} 
    \hline 
    $(n,p)$ & \multicolumn{6}{c|}{$(29,7)$} & \multicolumn{6}{c}{$(19,13)$} \\ \hline
    \diagbox[width=36pt]{Period}{$s$}
    	&0	&1	&2	&3	&4	&5	
     &0	&1	&2	&3	&4	&5 \\ \hline
    $l_s$	&1	&1	&1	&2	&2	&1 
    &12&6	&6	&4	&4	&4\\ \hline
    $k=1$	&1	&1	&2	&1	&1	&2	
    &1	&1	&2	&3	&3	&3\\ \hline
    $k=2$	&1	&1	&2	&2	&2	&2	 
    &1	&1	&12	&3	&12	&12\\ \hline
    $k=3$	&1	&1	&14	&14	&14	&2	
    &1	&1	&12	&12	&12	&12\\ \hline
    $k=4$	&1	&1	&98	&98	&98	&14 
    &1	&1	&156&12	&156	&156\\ \hline
    $k=5$	&1	&1	&686&686&686	&98	 
    &1	&1	&2028&156	&2028	&2028\\ \hline
    $k=6$	&1	&1	&4802&4802&4802	&686 
    &1	&1	&26364	&2028	&26364	&26364\\ 
    \hline
    \end{tabular}
\label{table:period:n29p7}
\end{table}

\section{Functional graph of Chebyshev Permutation Polynomial over ring $\mathbb{Z}_{p^k}$}
\label{sec:network}

Let $G_{p^k}$ denote the associated \emph{functional graph}
of a Chebyshev permutation polynomial over $\Zp{k}$. 
This graph can be constructed as follows: every number in $\mathcal{Z}_{p^k}$ is considered a vertex, and a vertex $x$ is directly linked to a vertex $y$ if and only if $y\equiv T_n(x)\pmod{p^k}$ \cite{cqli:network:TCASI2019}.
This section introduces some basic properties of $G_{p^k}$ and discloses its evolution rules concerning $k$.

\begin{figure*}[!htb]
    \centering
    \begin{minipage}[t]{0.65\twofigwidth}
    \centering
    \raisebox{0.15\twofigwidth}{\hfill
    \includegraphics[width=0.7\twofigwidth]{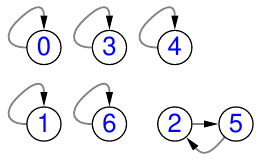}}
    a)
    \end{minipage} \hspace{4em}
    \begin{minipage}[t]{1.6\twofigwidth}
    \centering
    \includegraphics[width=2\twofigwidth]{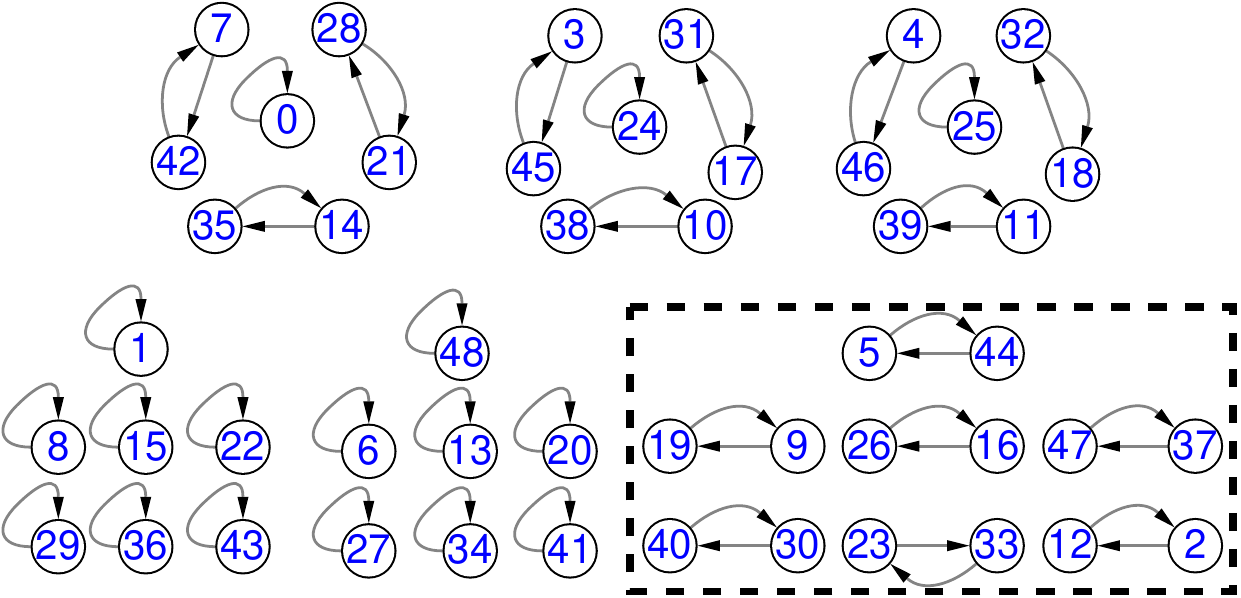}
    b)
    \end{minipage}\vspace{0.2em}
    \begin{minipage}{1.35\BigOneImW}
    \centering
    \includegraphics[width=1.35\BigOneImW]{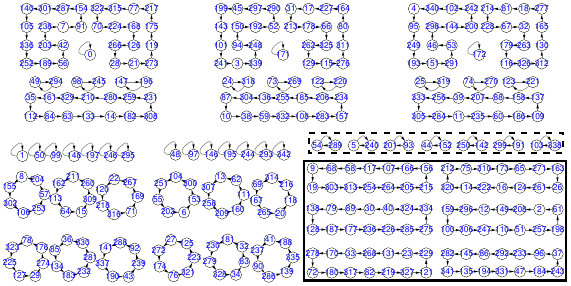}
    c)
    \end{minipage}
\caption{Functional graphs of Chebyshev polynomials over $\Zp{k}$ with $(n, p)=(13, 7)$: a) $k=1$; b) $k=2$; c) $k=3$.}
\label{fig:SMNcheby:n13p7}
\end{figure*}

\subsection{Some Basic Properties of $G_{p^k}$}

As a typical example, Fig.~\ref{fig:SMNcheby:n13p7} illustrates $G_{p^k}$ with $(n, p)=(13, 7)$ and $k\in\{1, 2, 3\}$, 
demonstrating some general rules of the functional graph
of a Chebyshev permutation polynomial over $\Zp{k}$:
\begin{itemize}[\newItemizeWidth]
    \item The graph $G_{p^k}$ is composed of some directed cycles without any transients, as shown in Proposition~\ref{prop:partcircle}.
    
    \item States $s$ and $p^k-s$ belong to cycles of the same length, as shown in Proposition~\ref{prop:symmetry}.
    
    \item There are some self-loops (i.e., a vertex has a directed edge that links to itself) for arbitrary parameter $k$, as shown in Proposition~\ref{prop:specialPoint2}.
    
    \item For each cycle of $G_p$, containing a node of value $s$, the functional graph $G_{p^k}$ contains a cycle of the same length, of which a node can be generated from $s$ in a fixed procedure when $k\ge 2$. Propositions~\ref{pro:Tnx=x:1} and \ref{pro:Tnx=x:2} discuss the existence of fixed point of $T_n(x)$ over $\Zp{k}$ for the two cases $T'_n(s)\equiv 1\pmod p$ and $T'_n(s)\not\equiv 1\pmod p$, respectively.
\end{itemize}

Note that Propositions~\ref{pro:Tnx=x:1} and \ref{pro:Tnx=x:2} can be attributed to the lifting theorem for the general polynomial congruence with prime power moduli given in \cite[Theorem 5.30]{Ap76:analyticnumbertheory}.
Additionally, there are some short cycles in $G_{p^k}$ for any parameter $k$, 
which may cause serious security pitfalls in applications using Chebyshev permutation polynomials
as production sources for pseudo-random number sequences and permutation vectors~\cite{Liu:Chebyrandom:2011, Chencc:chebyshev:TCASI2001}.

\begin{proposition}\label{prop:partcircle}
Any vertex of the functional graph of a Chebyshev permutation polynomial belongs to one and only one cycle, which composes a Chebyshev integer sequence.
\end{proposition}
\begin{proof}
The proof is straightforward according to the definitions of Chebyshev permutation polynomial and Chebyshev integer sequence.
\end{proof}

\begin{proposition}\label{prop:symmetry}
Given $s\in \mathcal{Z}_{p^k}$, if $T_n(s)\equiv s\pmod{p^k}$, then
$T_n(p^k-s)\equiv (p^k-s)\pmod{p^k}$,
where $k\in\mathbb{N}^+$.
\end{proposition}
\begin{proof}
From \eqref{cheby:power}, one has $T_n(-s)=-T_n(s)$ since $n$ is odd. 
Thus, $T_n(p^k-s)=-T_n(s-p^k)$, which implies $T_n(p^k-s)\equiv-T_n(s)\equiv (p^k-s)\pmod{p^k}$.
\end{proof}

\begin{proposition}\label{prop:specialPoint2}
The Chebyshev permutation polynomial satisfies
$T_n(s)\equiv s\pmod{p^k}$ 
for $s\in\{0, 1, (p^k+1)/2, (p^k-1)/2, p^k-1\}$,
where $k\in\mathbb{N}^+$.
\end{proposition}
\begin{proof}
When $s\in\{0, 1\}$, it follows from \eqref{eq:recursive} that $T_n(s)=s$ for 
any $n\ge 1$.
Furthermore, from \eqref{eq:recursive}, the sequence $\{T_i(\frac{1}{2}) \}_{i=1}^{\infty}$ is a sequence of period six:
$\left\{\frac{1}{2}, \frac{-1}{2}, -1, \frac{-1}{2}, \frac{1}{2}, 1, \frac{1}{2}, \frac{-1}{2}, \cdots\right\}.$
Given that $n\equiv \pm 1\pmod 6$, one deduces that $T_n(\frac{1}{2})=\frac{1}{2}$.
As $n\cdot \frac{(n-m-1)!}{m!(n-2m)!}=\binom{n-m}{m}+\binom{n-m-1}{m-1}$,
this expression is an integer. Hence, \eqref{cheby:power} can be represented as
$T_n(x)=\sum_{i=1}^n a_i x^i$,
where $2^{i-1}\mid a_i$.
Moreover, one has
\(
T_n\left(\frac{p^k+ 1}{2} \right)=\frac{p^k+ 1}{2} \sum_{i=1}^n \frac{a_i}{2^{i-1}}\left(p^k+ 1\right)^{i-1}
\equiv \frac{p^k+ 1}{2} \sum_{i=1}^n \frac{a_i}{2^{i-1}}\pmod{p^k}.
\)
Since $\sum_{i=1}^n \frac{a_i}{2^{i-1}}=2\cdot T_n(\frac{1}{2})=1$,
it follows that
$T_n(\frac{p^k+ 1}{2})\equiv \frac{p^k+ 1}{2}\pmod{p^k}$.
Thus, one obtains
$T_n\left(\frac{p^k-1}{2}\right) \equiv \frac{p^k-1}{2}\pmod{p^k}$ and $T_n(p^k-1) \equiv (p^k-1)\pmod{p^k}$ both from Proposition \ref{prop:symmetry}.
\end{proof}

\begin{proposition}\label{pro:Tnx=x:1}
Given $s\in\mathcal{Z}_p$, if $T_n(s)\equiv s\pmod p$ and $T'_n(s)\equiv 1\pmod p$, then there is an element in $\mathcal{Z}_{p^k}$, $s_k$, satisfying
\begin{equation}\label{condition}
x\equiv s\pmod p \text{~and~} x\equiv T_n(x)\pmod{p^k},
\end{equation}
where $s_k=s$ if $k\leq z_s+1$; $s_k$ is generated by calculating $s_j=s+\sum_{i=z_s+1}^{j-1} y_i p^{i-z_s}$ and $y_i=(\alpha\cdot\frac{T_n(s_{i})-s_{i}}{p^{i}})\bmod p$
for $j=z_s+2, \cdots, k$ in order if $k>z_s+1$, where
$z_s=\max\{e \vvert T'_n(s) \equiv 1\pmod{p^e}\}$ and $\alpha$ is the inverse of 
$\frac{1-T'_n(s)}{p^{z_s}}$ in the field $\mathbb{F}_p$.
\end{proposition}
\begin{proof}
The proposition is proved via induction on $k$.
By the definition of $z_s$, one has
\begin{equation}\label{eq:Temp1}
\begin{cases}
    T'_n(s)\equiv 1\pmod{p^{z_s}},\\
    T'_n(s)\not\equiv 1\pmod{p^{z_s+1}},
\end{cases}
\end{equation}
namely
\begin{equation}\label{eq:T'nx:d}
T'_n(s)=1+c\cdot p^{z_s}, 
\end{equation}
where $c\nmid p$.
Combining $T_n(s)\equiv s\pmod p$, condition~\eqref{eq:Temp1}, and Proposition~\ref{lemma:Tn(x):w} with $(w, h)=(z_s, 0)$, one deduces that 
$T_n(s)\equiv s\pmod{p^k}$ for $k\leq z_s+1$.
Thus, the proposition holds when $k\leq {z_s}+1$.
Assume that the proposition holds when $k=e>{z_s}+1$, i.e., $T_n(s_e)\equiv s_e\pmod{p^e}$, where
$s_e=s+\sum_{i={z_s}+1}^{e-1} y_i p^{i-z_s}$.
Define
\begin{equation}\label{eq:TnXs:2}
T_n(s_e)=s_e+h_ep^e.
\end{equation}
By condition~\eqref{eq:Temp1} and Lemma~\ref{lemma:Tn:0pw} with 
$(w, h)=(z_s, \sum_{i=z_s+1}^{e-1}y_ip^{i-z_s-1})$, one has
\begin{equation}\label{eq:TnXs:pw1}
T'_n(s_e) \equiv T'_n(s)\pmod{p^{z_s+1}}
\end{equation}
and
\begin{equation}\label{eq:TnXs:pw2}
\frac{T^{(m)}_n(s_e)\cdot p^m}{m!} \equiv 0\pmod{p^{z_s+2}}
\end{equation}
for $m\ge2$.
Applying Taylor's formula to $T_n(x)$ at $s_e$ and
referring to \eqref{eq:TnXs:2}, one obtains
\begin{align*}
T_n(s_e+&yp^{e-z_s})\\
        &=s_e+h_e p^e+\sum^n_{m=1}\frac{(y p^{e-z_s-1}\cdot p)^m}{m!}T^{(m)}_n(s_e),
\end{align*}
where $y\in \mathbb{Z}$.
Since $v_p\left(\frac{(y p^{e-z_s-1}\cdot p)^m}{m!}T^{(m)}_n(s_e)\right)\geq (e-z_s-1)m+z_s+2\ge (e-z_s-1)2+z_s+2=2e-z_s>e+1$ for $m\ge 2$, substituting congruence~\eqref{eq:TnXs:pw2} into the preceding equation, one obtains
$T_n(s_e+yp^{e-z_s})\equiv s_e+h_ep^e+yp^{e-z_s}T'_n(s_e)\pmod{p^{e+1}}$.
Incorporating congruence~\eqref{eq:TnXs:pw1} and
\eqref{eq:T'nx:d} into this congruence, one obtains
\begin{equation*}
T_n(s_e+y p^{e-z_s})\equiv s_e+yp^{e-z_s}+(h_e+cy)p^e\pmod{p^{e+1}}.
\end{equation*} 
Substituting 
$y_e=
(\alpha\cdot \frac{T_n(s_e)-s_e}{p^e})\bmod p$ for $y$, one derives
$h_e+cy\equiv 0\pmod{p}$ from \eqref{eq:T'nx:d} and \eqref{eq:TnXs:2}.
Therefore, $s_{e+1}=s_e+y_ep^{e-z_s}$ is an element of $\mathcal{Z}_{p^{e+1}}$ that satisfies $T_n(x)\equiv x\pmod{p^k}$ for $k=e+1$.
Thus, the proposition holds when $k=e+1>z_s+2$,
completing the induction and the proof of the proposition.
\end{proof}

\begin{proposition}\label{pro:Tnx=x:2}
Given $s\in\mathcal{Z}_p$, if $T_n(s)\equiv s\pmod p$ and $T'_n(s)\not\equiv 1\pmod p$, then $s_k$ is the sole element in $\mathcal{Z}_{p^k}$ satisfying 
\begin{equation*}
x\equiv s\pmod p \text{~and~} x\equiv T_n(x)\pmod{p^k},
\end{equation*}
where $s_k$ is generated by calculating $s_j=s+\sum_{i=1}^{j-1} y_{i} p^i$ and $y_i=(\beta\cdot\frac{T_n(s_i)-s_i}{p^i})\bmod p$ for $j=1, 2, \cdots, k$ in order, $k\ge 1$, 
and $\beta$ is the inverse of $(1-T'_n(s))$ in the field $\mathbb{F}_p$.
\end{proposition}
\begin{proof}
The proposition is proved by mathematical induction on $k$.
When $k=1$, the proposition holds by $s_1=s$.
Assume that the proposition holds for $k=e>1$, namely $T_n(s_e)\equiv s_e\pmod{p^e}$, 
where $s_e=s+\sum_{i=1}^{e-1}y_ip^i$ is the unique element in $\mathcal{Z}_{p^e}$ satisfying condition~(\ref{condition}).
Let
$T_n(s_e)=s_e+h_e p^e$, and $T'_n(s)=d+c\cdot p$,
where $d\not\equiv 1\pmod p$.
Using Taylor's formula, one obtains
\(
T_n(s_e+y p^e)
\equiv T_n(s_e)+y p^e\cdot T'_n(s_e)\pmod{p^{e+1}},
\)
where $y\in \mathbb{Z}$.
Since $T'_n(s_e)=T'_n(s+\sum_{i=1}^{e-1}y_ip^i)\equiv T'_n(s)\pmod{p}$, one can further derive that
\(
T_n(s_e+y p^e)
\equiv T_n(s_e)+y p^e\cdot T'_n(s)\pmod{p^{e+1}}.
\)
Incorporating the expressions of $T_n(s_e)$ and $T'_n(s)$ into
this congruence, one has
\begin{equation*}
\begin{aligned}
T_n(s_e+y p^e)
    &\equiv s_e+h_e p^e+y p^e\cdot(d+c\cdot p)\pmod{p^{e+1}}\\
    &\equiv s_e+y p^e+(h_e-(1-d)y)p^e\pmod{p^{e+1}}.
\end{aligned}
\end{equation*}
Substituting $y_e=((1-d)^{-1}\cdot h_e)\bmod p$ into $y$ in the preceding congruence, one can get
\begin{equation}\label{eq:hed}
(h_e-(1-d)y)\equiv 0\pmod{p},
\end{equation}
where $(1-d)^{-1}=(1-T'_n(s)+c\cdot p)^{-1}\equiv \beta\pmod{p}$.
Since $y_e$ is the unique value of $y$ 
in $\mathcal{Z}_p$ that satisfies congruence~\eqref{eq:hed},
it follows that 
$s_{e+1}=s_e+y_e\cdot p^e$ is the unique element in $\mathcal{Z}_{p^{e+1}}$
satisfying condition~(\ref{condition}) for $k=e+1$.
Therefore, the proposition holds for $k=e+1$, 
completing the induction and proving the proposition.
\end{proof}

\subsection{Evolution Rules of $G_{p^k}$ Concerning $k$}

This subsection unveils the elegant patterns in the cycle distribution of $G_{p^k}$ with increasing $k$, as observed from numerous experiments on Chebyshev permutation polynomials with random parameters. As presented in Table~\ref{table:cycle:n13p7}, the number of cycles with period $T_c$ in $G_{p^k}$ equals that with period $p \cdot T_c$ in $G_{p^{k+1}}$ when $T_c\ge p^2$, which is summarized in Theorem~\ref{pro:Threshold:Tc}.

\setlength\tabcolsep{4.8pt} 
\addtolength{\abovecaptionskip}{-1pt}
\begin{table}[!htb]
    \centering
    \caption{The number of cycles of period $T_c$ in $G_{p^k}$ when $(n, p)=(13, 7)$.}
    \begin{tabular}{*{10}{c|}c}
    \hline
    \diagbox{$k$}{$C_{T_c, k}$}{$T_c$}
        &1	&2	&7	&14	&49	&98	&343&686&2401&4802 \\ \hline
    1	&5	&1	&	&	&	&	&	&	&	&\\ \hline
    2	&17	&16	&	&	&	&	&	&	&	&\\ \hline
    3	&17	&16	&12	&15	&	&	&	&	&	&\\ \hline
    4	&17	&16	&12	&15	&12	&15	&	&	&	&\\ \hline
    5	&17	&16	&12	&15	&12	&15	&12	&15	&	&\\ \hline
    6	&17	&16	&12	&15	&12	&15	&12	&15	&12	&15\\ \hline
    \end{tabular}
\label{table:cycle:n13p7}
\end{table}

\begin{figure*}[!htb]
    \centering
    \begin{minipage}{1.8\twofigwidth}
      \centering
      \includegraphics[width=1.8\twofigwidth]{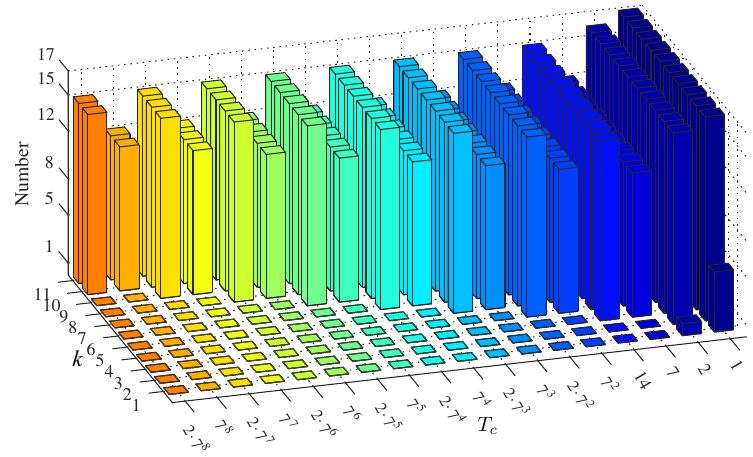}
      a)
    \end{minipage}
    \begin{minipage}{1.8\twofigwidth}
      \centering
      \includegraphics[width=1.8\twofigwidth]{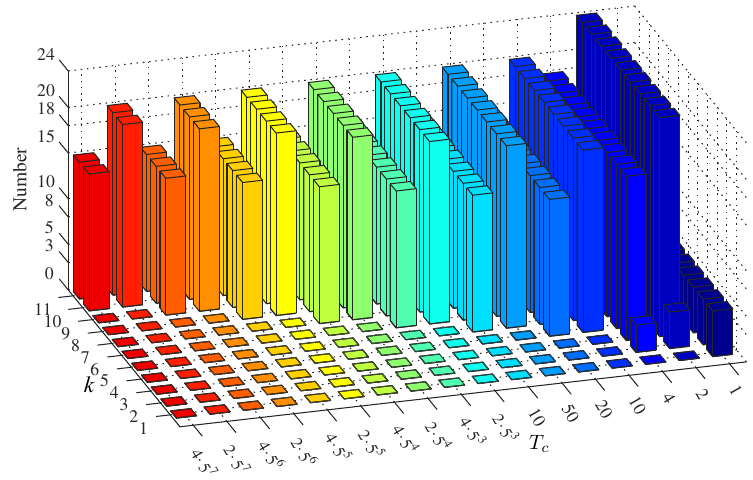}
      b)
    \end{minipage}\\
    \begin{minipage}{1.8\twofigwidth}
      \centering
      \includegraphics[width=1.8\twofigwidth]{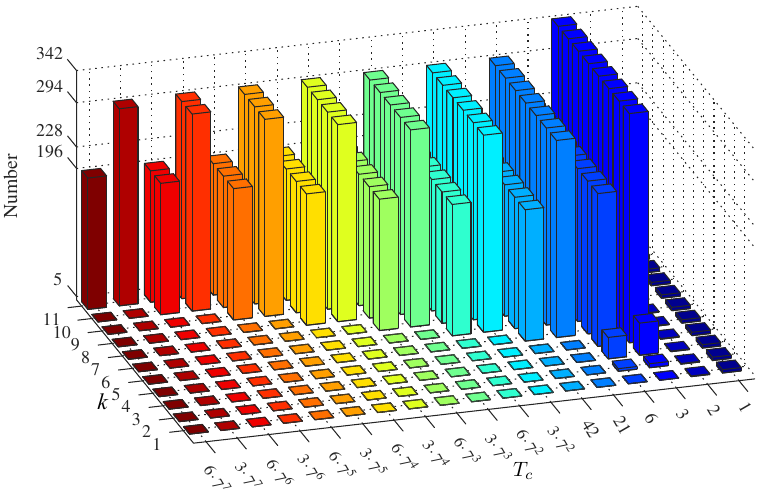}
      c)
    \end{minipage}
    \begin{minipage}{1.8\twofigwidth}
      \centering
      \includegraphics[width=1.8\twofigwidth]{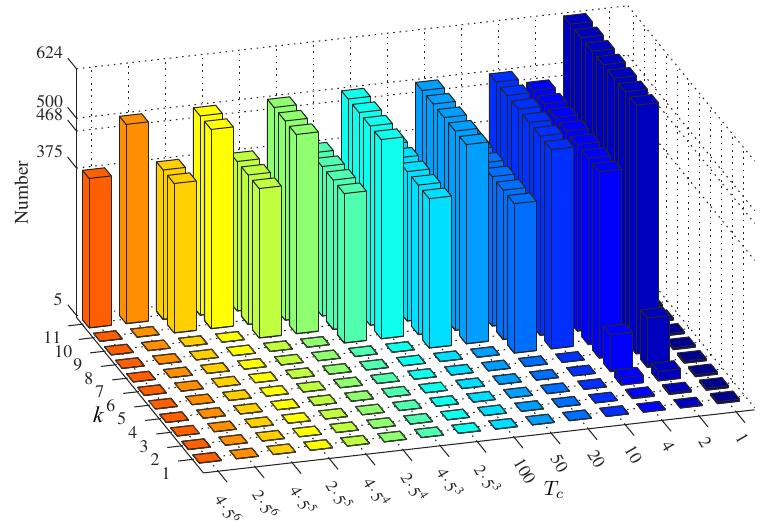}
      d)
    \end{minipage}
\caption{The cycle distribution of Chebyshev polynomial $T_n(x)$ over $\Zp{k}$, where $k=1, 2, \cdots, 11$: a) $(n, p)=(13, 7)$; b) $(n, p)=(43, 5)$; c) $(n, p)=(19, 7)$; d) $(n, p)=(443, 5)$.}
\label{fig:cycle:disall}
\end{figure*}

\begin{theorem}\label{pro:Threshold:Tc}
Let $C_{T_c, k}$ denote the number of cycles of length $T_c$
in the functional graph of Chebyshev permutation polynomials over $\Zp{k}$. Then, one has
\begin{equation*}
C_{p\cdot T_c, k+1}=C_{T_c, k}
\end{equation*}
when $T_c\geq p^2$.
\end{theorem}
\begin{proof}
From Proposition~\ref{prop:partcircle}, one can deduce that a directed cycle of length $T_c$ in the functional graph $G_{p^k}$ corresponds to a Chebyshev integer sequence $\{T_n^i(s)\bmod p^k\}_{i\ge 0}$ with least period $T_c$, where $s\in\mathcal{Z}_{p^k}$. According to the definition of $l_s$ in Theorem~\ref{Theorem:periodValue}, one finds that $l_s<p$. Referring also to the definition of $N_s$ and applying Fermat's Little Theorem, one obtains $N_s\leq p$. When $T_c\geq p^2$, assume that $k\leq v=\max\{e \vvert T_n^{N_s\cdot l_s}(s)\equiv s\pmod{p^e}\}$. Consequently, one has $T_n^{N_s\cdot l_s}(s)\equiv s\pmod{p^k}$, which implies that $T_c\leq N_s\cdot l_s<p^2$. However, this conclusion contradicts the assumption. Thus, one must have $k>v$.

By Theorem~\ref{Theorem:periodValue}, any cycle of length $T_c=N_s\cdot l_s\cdot p^{k-v}$
is ``expanded" to a cycle of length $N_s\cdot l_s\cdot p^{k+1-v}=T_c\cdot p$ over $\Zp{k+1}$.
Now, consider two cycles of length $T_c<p^2$ from $G_{p^k}$, denoted by $C_1=\{T^i_n(s)\bmod p^k\}^{T_c}_{i=0}$, and $C_2=\{T^i_n(s')\bmod p^k\}^{T'_c}_{i=0}$.
If $C_1$ and $C_2$ are expanded in the aforementioned manner
to the same cycle in $G_{p^{k+1}}$, there must exist an integer $j$ such that $T^j_n(s)\equiv s'\pmod {p^{k+1}}$.
Hence, a cycle of length $T_c \cdot p$ over $\Zp{k+1}$ corresponds to precisely one cycle of length $T_c$ over $\Zp{k}$. Therefore, one concludes that $C_{p\cdot T_c, k+1}=C_{T_c, k}$ when $T_c\geq p^2$.
\end{proof}

Lemma~\ref{lemma:thershold:equal} describes
that the derivative of $T_{n^{N_s l_s}}$
is congruent to $\pm n^{N_s l_s}$ modulo $p^w$ under specific conditions.
Lemma~\ref{lemma:nk1k2} offers a sufficient condition for deducing $n^{k_1}\equiv 1\pmod{p^w}$ from $n^{k_1 k_2}\equiv 1\pmod{p^w}$.
Lemma~\ref{lemma:ls=lsi} reveals that the value of $l_x$ remain constant for any $x \in \{T^i_n(s)\bmod p\}^{N_s-1}_{i=0}$.
By applying Lemmas~\ref{lemma:thershold:equal} and \ref{lemma:nk1k2},
one can infer that the maximum values of $e$, satisfying $T'_{n^{N_{s}l_{s}}}(s) \equiv 1\pmod{p^e}$, are identical for any
$s\in\mathcal{Z}_p$, and depend solely on the pair $(n, p)$ as shown in Lemma~\ref{pro:thershold:equal}.
From Proposition~\ref{prop:partcircle}, one can deduce that the functional graph $G_p$ consists of
several cycles with varying lengths, each corresponding to a sequence $\{T^i_n(s)\bmod p\}^{N_s-1}_{i=0}$ for some $s\in\Zp{}$.
Since each cycle in $\Zp{}$ evolves independently as $k$ increases, one can 
focus on the part of $G_{p^k}$ composed of vertices with values in the set $\{x \vvert x\in \mathcal{Z}_{p^k}, (x\bmod p)\in\{T_n^i(s)\bmod p\}_{i=0}^{N_s-1}\}$.
Proposition~\ref{pro:SMN:+1} presents the exact number of cycles of various lengths.

\begin{lemma}\label{lemma:thershold:equal}
Given $s\in\mathcal{Z}_p\backslash \{1, p-1\}$, if $T_{n^{N_sl_s}}(s)\equiv s\pmod{p^w}$, then $T'_{n^{N_sl_s}}(s)\equiv (h\cdot n^{N_s})^{l_s}\pmod{p^w}$,
where $h\in\{-1, 1\}$ satisfying $T'_{n^{N_s}}(s)\equiv h\cdot n^{N_s}\pmod{p}$ 
and $w\in\mathbb{N}^+$.
\end{lemma}
\begin{proof}
Referring to \cite[Eq. (10)]{Rayes:Factorization:CM05}, one has 
\begin{equation}\label{eq:U_nm_s}
U_{m_1m_2-1}(s)=U_{m_1-1}(T_{m_2}(s))\cdot U_{m_2-1}(s),
\end{equation}
where $m_1, m_2\in\mathbb{N}^+$.
Setting $(m_1, m_2)=(n^{N_s(l_s-1)}, n^{N_s})$ in this equation, one obtains
\[
U_{n^{N_sl_s}-1}(s)=U_{n^{N_s(l_s-1)}-1}(T_{n^{N_s}}(s))\cdot U_{n^{N_s}-1}(s).
\]
By the definition of $N_s$, one has $T_{n^{N_s}}(s)\equiv s\pmod p$,
yielding
$U_{n^{N_sl_s}-1}(s) \equiv U_{n^{N_s(l_s-1)}-1}(s)\cdot U_{n^{N_s}-1}(s)\pmod{p}$.
Setting $(m_1, m_2)=(n^{N_s(l_s-2)}, n^{N_s})$ in \eqref{eq:U_nm_s} again, one has
\[
U_{n^{N_s(l_s-1)}-1}(s)=U_{n^{N_s(l_s-2)}-1}(T_{n^{N_s}}(s))\cdot U_{n^{N_s}-1}(s).
\]
Thus, one deduce that
\(
U_{n^{N_sl_s}-1}(s) \equiv U_{n^{N_s(l_s-2)}-1}(s)\cdot (U_{n^{N_s}-1}(s))^2\pmod{p}.
\)
Repeating this process $l_s-3$ more times, one finds
\begin{equation}\label{eq:U=Ulx}
U_{n^{N_sl_s}-1}(s)\equiv (U_{n^{N_s}-1}(s))^{l_s}\pmod{p}.
\end{equation}

Replacing $n$ in Lemma~\ref{lemma:T'n(x):n} with $w=1$ by $n^{N_s}$
and noting that $T_{n^{N_s}}(s)\equiv s\pmod p$, it follows that
$T'_{n^{N_s}}(s)\equiv \pm n^{N_s}\pmod{p}$.
Replacing $n$ in
\eqref{eq:relation} by $n^{N_s}$, one obtains
$T'_{n^{N_s}}(s)=n^{N_s}U_{n^{N_s}-1}(s)$.
If $T'_{n^{N_s}}(s)\equiv -n^{N_s}\pmod{p}$, it follows that
$U_{n^{N_s}-1}(s)\equiv -1\pmod{p}$ since $\gcd(n, p)=1$.
Applying \eqref{eq:U=Ulx}, one then has
\begin{equation}\label{eq:unNsls:pm1}
U_{n^{N_sl_s}-1}(s) \equiv (-1)^{l_s}\pmod p.
\end{equation}
Using Lemma~\ref{lemma:T'n(x):n} again by replacing $n$ with $n^{N_sl_s}$ and 
noting
$T_{n^{N_sl_s}}(s)\equiv s\pmod{p^w}$, one also has
$U_{n^{N_sl_s}-1}(s)\equiv \pm 1\pmod{p^w}$. 
From \eqref{eq:unNsls:pm1}, it follows that 
\(
 U_{n^{N_sl_s}-1}(s)\equiv (-1)^{l_s}\pmod{p^w}.
\)
Thus, from \eqref{eq:relation}, one concludes that
\(
T'_{n^{N_sl_s}}(s)\equiv(-n^{N_s})^{l_s}\pmod{p^w},
\)
which proves the case $h=-1$ of this lemma.
The proof for the case $h=1$ is analogous and is omitted.
\end{proof}

\begin{lemma}\label{lemma:nk1k2}
Let $x\in \mathbb{Z}$, $k_1, k_2, w\in \mathbb{N}^+$, and $p$ be a prime such that $p\nmid k_2$.
If $x^{k_1 k_2} \equiv 1 \pmod{p^w}$ and $x^{k_1}\equiv 1 \pmod{p}$, then $x^{k_1} \equiv 1 \pmod{p^w}$.
\end{lemma}
\begin{proof}
From $x^{k_1}\equiv 1\pmod p$, one has 
$x^{k_1}=1+c\cdot p$.
Thus,
\(
x^{k_1k_2}=(1+c\cdot p)^{k_2}
=1+\sum_{i=1}^{k_2}\binom{k_2}{i}(c\cdot p)^i 
=1+c\cdot p\left(k_2+\sum_{i=2}^{k_2}\binom{k_2}{i}(c\cdot p)^{i-1} \right).
\)
It further deduces from $x^{k_1k_2}\equiv 1\pmod{p^w}$ that $p^w\mid ((c\cdot p)\cdot(k_2+\sum_{i=2}^{k_2}\binom{k_2}{i}(c\cdot p)^{i-1}))$.
Since $p\nmid k_2$ and $p\mid \sum_{i=2}^{k_2}\binom{k_2}{i}(c\cdot p)^{i-1}$, one has $p^w\nmid (k_2+\sum_{i=2}^{k_2}\binom{k_2}{i}\cdot (c\cdot p)^{i-1})$. Therefore, $p^w\mid (c\cdot p)$, implying $x^{k_1}\equiv 1\pmod{p^w}$ from the above expression of $x^{k_1}$.
\end{proof}

\begin{lemma}\label{lemma:ls=lsi}
Given $s\in\mathcal{Z}_p$, one has $l_{c_i}=l_s$,
where $c_i=T_n^i(s)\bmod{p}$ and $i\in\{0,1,\cdots, N_s-1\}$.
\end{lemma}
\begin{proof}
If $N_s=1$, one has $c_i=s$ for $i=0$, which means the lemma holds in this case. 

Now, suppose $N_s \neq 1$. From Proposition~\ref{prop:specialPoint2} with $k=1$, one deduces that $s \notin \{1, p-1\}$. Furthermore, from $T_{n^{N_s}}(s) \equiv s \pmod{p}$ and \eqref{eq:U}, one concludes that $U_{n^{N_s}-1}(s) \equiv \pm 1 \pmod{p}$. 

Consider first the case where $U_{n^{N_s}-1}(s) \equiv 1 \pmod{p}$. By setting $(m_1, m_2) = (n^i, n^{N_s})$ in \eqref{eq:U_nm_s}, one obtains the relation
$U_{n^{i+N_s}-1}(s)=U_{n^i-1}(T_{n^{N_s}}(s))\cdot U_{n^{N_s}-1}(s)$,
which implies that
$U_{n^{i+N_s}-1}(s)\equiv U_{n^i-1}(s)\pmod{p}$.
Similarly, by setting $(m_1, m_2) = (n^{N_s}, n^i)$ in \eqref{eq:U_nm_s}, one finds
$U_{n^{i+N_s}-1}(s)=U_{n^{N_s}-1}(T_{n^i}(s))\cdot U_{n^i-1}(s)$.
Thus, it follows that
\[U_{n^i-1}(s) \equiv U_{n^{N_s}-1}(T_{n^i}(s)) \cdot U_{n^i-1}(s) \pmod{p}.
\]
Since $s \notin \{1, p-1\}$, one concludes that $c_i \notin \{1, p-1\}$ by applying Proposition~\ref{prop:specialPoint2} with $(n^i, 1)$ in place of $(n, k)$. This implies that $p \nmid (s^2 - 1)$ and $p \nmid T^2_{n^i}(s) - 1$. Replacing $n$ with $n^i$ in \eqref{eq:Tn2:unx}, one finds that $U_{n^i-1}(s) \not\equiv 0 \pmod{p}$. Consequently, $U_{n^{N_s}-1}(T_{n^i}(s)) \equiv 1 \pmod{p}$.
From \eqref{eq:relation} and the definition of $N_x$, one obtains
$T'_{n^{N_{c_i}}}(c_i) = n^{N_s} \cdot U_{n^{N_s}-1}(c_i)$.
Therefore, one has
$T'_{n^{N_{c_i}}}(c_i) \equiv n^{N_s} \cdot U_{n^{N_s}-1}(T_{n^i}(s)) \equiv n^{N_s} \pmod{p}$.
Recalling the definition of $l_x$, one concludes that $l_{c_i} = l_s$ for all $i \in \{0, 1, \cdots, N_s-1\}$.
The proof for the case where $U_{n^{N_s}-1}(s) \equiv -1 \pmod{p}$ proceeds similarly and is omitted.
\end{proof}

\begin{lemma}\label{pro:thershold:equal}
For any $s\in\mathcal{Z}_p$, one has
$\max\{e \vvert T'_{n^{N_{s}l_{s}}}(s)\equiv 1\pmod{p^e}\}=\max\{e \vvert n^{2\cdot\ord(n^2)} \equiv 1\pmod{p^e}\}$.
\end{lemma}
\begin{proof}
Let $w_s=\max\{e \vvert T'_{n^{N_s l_s}}(s) \equiv 1\pmod{p^e}\}$, so one has
$T'_{n^{N_sl_s}}(s) \equiv 1\pmod{p^{w_s}}$.
From the definition of $N_s$, it follows that
$T_{n^{N_sl_s}}(s) \equiv s\pmod{p}$. 
By substituting $(n, h)=(n^{N_s l_s}, 0)$ in Proposition~\ref{lemma:Tn(x):w}, one obtains
\begin{equation}\label{eq:TNs:pws+1}
T_{n^{N_sl_s}}(s)\equiv s\pmod{p^{w_s+1}}.
\end{equation}

Now, suppose $s \in \{1, p-1\}$. One then observes that $N_s = 1$ from Proposition~\ref{prop:specialPoint2}.
From the definition of $l_x$ and \eqref{eq:pm12},
one sees that both
$l_1$ and $l_{p-1}$ are equal to $\ord(n^2)$.
Replacing $n$ in Lemma~\ref{lemma:tn(-1+p):equal} by $n^{l_1}$, one notes that
$T'_{n^{l_1}}(p-1)\equiv 1\pmod{p^w}$ is equivalent to 
$T'_{n^{l_1}}(-1)\equiv 1\pmod{p^w}$ for any positive integer $w$.
Thus, 
$T'_{n^{l_{p-1}}}(p-1)\equiv 1\pmod{p^w}\Leftrightarrow T'_{n^{l_1}}(1)\equiv 1\pmod{p^w}$ since $l_1=l_{p-1}$.
Setting $(n, m)=(n^{l_1}, 1)$ in \eqref{eq:t'n1} and \eqref{eq:pm12}, one deduces that
$T'_{n^{l_1}}(\pm 1)=n^{2l_1}$.
Hence, the two preceding equivalent congruences are reduced to 
\(n^{2\cdot \ord(n^2)}\equiv 1\pmod{p^w}\).
Thus, the lemma holds for $s\in\{1, p-1\}$.

Now suppose $s\notin\{1, p-1\}$.
From \eqref{eq:TNs:pws+1} and Lemma~\ref{lemma:thershold:equal} with $w=w_s+1$,
one has $T'_{n^{N_sl_s}}(s)\equiv (h\cdot n^{N_s})^{l_s}\pmod{p^{w_s+1}}$,
where $h\in\{-1, 1\}$ satisfies $T'_{n^{N_s}}(s)\equiv h\cdot n^{N_s}\pmod{p}$.
Assuming \( h = -1 \), it follows that
\begin{equation}\label{eq:nNsls:pws+1}
T'_{n^{N_sl_s}}(s)\equiv (-n^{N_s})^{l_s}\pmod{p^{w_s+1}},
\end{equation}
and $l_s=\ord(T'_{n^{N_s}}(s))=\ord(-n^{N_s})$.
Let $l=\ord(n^{N_s})$. Then, $(n^{N_s})^l\equiv 1\pmod p$ and $(-n^{N_s})^{2l} \equiv 1\pmod p$. Hence, $l_s\mid 2l$.
According to \cite[Theorem 1.15]{Rudolf:Introfinite:1994}, one obtains 
$l=\frac{\ord(n)}{\gcd(N_s, \ord(n))}$
and $\ord(n^2)=\frac{\ord(n)}{\gcd(2, \ord(n))}$.
Thus, 
\begin{equation}\label{eq:l3}
l=\frac{r_1\cdot \ord(n^2)}{r_2},
\end{equation}
where $r_1=\gcd(2, \ord(n))$ and $r_2=\gcd(N_s, \ord(n))$.
Comparing $T'_{n^{N_sl_s}}(s) \equiv 1\pmod{p^{w_s}}$ and \eqref{eq:nNsls:pws+1}, one obtains 
\begin{equation}\label{eq:nspw}
(-n^{N_s})^{l_s} \equiv 1\pmod{p^{w_s}}.
\end{equation}
Since \( l_s \mid 2l \), it follows that
$(n^{N_s})^{2l} \equiv 1\pmod{p^{w_s}}$,
which implies
$n^{2\cdot \ord(n^2)\cdot \frac{r_1\cdot N_s}{r_2}} \equiv 1\pmod{p^{w_s}}$.
Given that $n^{2 \cdot \ord(n^2)}\equiv 1\pmod p$, $p\nmid \frac{r_1\cdot N_s}{r_2}$ and Lemma~\ref{lemma:nk1k2}, 
one concludes that $^{2\cdot \ord(n^2)}\equiv 1\pmod{p^{w_s}}$.
Hence, $\max\{e \vvert n^{2\cdot\ord(n^2)} \equiv 1\pmod{p^e} \}\ge w_s$.

Assuming the equality in this relational expression does not hold, one has $n^{2\cdot \ord(n^2)}\equiv 1\pmod{p^{w_s+1}}$.
Since $2\cdot \ord(n^2)$ always divides $2r_2\cdot l$ from \eqref{eq:l3} and $r_2\mid N_s$, one has $2\cdot \ord(n^2)\mid 2N_sl$, implying
$n^{2N_s l}=(-n^{N_s})^{2l} \equiv 1\pmod{p^{w_s+1}}$.
Since $l_s \mid 2l$ and $p\nmid l$, one concludes that
$2l=l_s\cdot h$ and $p\nmid h$.
Setting $(x, k_1, k_2, w)=(-n^{N_s}, l_s, h, w_s)$ in
Lemma~\ref{lemma:nk1k2}, one obtains $(-n^{N_s})^{l_s} \equiv1\pmod{p^{w_s+1}}$ from
$(-n^{N_s})^{l_s h}\equiv 1\pmod{p^{w_s+1}}$ and \eqref{eq:nspw}.
Thus, it follows from \eqref{eq:nNsls:pws+1} that $T'_{n^{N_sl_s}}(s) \equiv1\pmod{p^{w_s+1}}$,
which contradicts the definition of \( w_s \).
Hence, $\max\{e \vvert n^{2\cdot\ord(n^2)} \equiv 1\pmod{p^e} \}=w_s$.
The proof for the case $h=1$ is similar and is omitted.
\end{proof}

\begin{proposition}
\label{pro:SMN:+1}
Given $s\in\mathcal{Z}_p$, the elements of the set $\{x \vvert 
x\in \mathcal{Z}_{p^k}, (x\bmod p)\in\{T_n^i(s)\bmod p\}_{i=0}^{N_s-1}\}$ form cycles in the functional graph $G_{p^k}$ with three distinct presentations:
one cycle of length $N_s$ and
$\frac{p^{\min(k-1, w)}-1}{l_s}$ cycles of length $N_sl_s$;
$\frac{(p-1)p^{w-1}}{l_s}$ cycles of length $N_sl_sp^{k-w-t}$
for each $t\in \{1, 2, \cdots, k-1-w\}$, available only if $k-1-w>0$, where 
$w=\max\{e \vvert n^{2\cdot\ord(n^2)} \equiv 1\pmod{p^e}\}$.
\end{proposition}
\begin{proof}
Let $c_i=T_n^i(s)\bmod p$ for $i=0,1, \cdots, N_s-1$.
By this definition and that of $N_x$ and $l_x$, it follows that $N_{c_i}=N_s$;
$l_{c_i}=l_s$ from Lemma~\ref{lemma:ls=lsi}.
Using the definition of $w$ along with Lemma~\ref{pro:thershold:equal},
one has
$w=\max\{e \vvert T'_{n^{N_x l_x}}(x)\equiv 1\pmod{p^e}\}$
for any $x\in\mathcal{Z}_p$.
Then, one further obtains the following system of congruences:
\begin{equation}\label{eq:tn'x:change}
\begin{cases}
T'_{n^{N_sl_s}}(c_i)\equiv 1\pmod{p^{w}};\\
T'_{n^{N_sl_s}}(c_i)\not\equiv 1\pmod{p^{w+1}}.
\end{cases}
\end{equation}
Since $T_{n^{N_s}}(c_i)\equiv c_i\pmod p$, one deduces $T_{n^{N_sl_s}}(c_i)\equiv c_i\pmod p$.
Replacing $n$ by $n^{N_sl_s}$ in Proposition~\ref{lemma:Tn(x):w} with $h=0$, 
one concludes that 
$T_{n^{N_sl_s}}(c_i)\equiv c_i\pmod{p^{w+1}}$, which
implies
$T_{n^{N_sl_s}}(c_i)\equiv c_i\pmod{p^{w}}$. 
Referring to Lemma~\ref{lemma:Tn:0pw}, one arrives at congruences
\begin{equation}\label{eq:TnXk:pw1}
T'_{n^{N_sl_s}}(a_i)\equiv T'_{n^{N_sl_s}}(c_i )\pmod{p^{w+1}}
\end{equation}
and
\begin{equation}\label{eq:TnXk:pw2}
\frac{T^{(m)}_{n^{N_sl_s}}(a_i)\cdot p^m}{m!} \equiv 0\pmod{p^{w+2}},
\end{equation}
where $m\ge2$ and $a_i \equiv c_i\pmod p$.
Applying Taylor's expansion, one can write
\[
T_{n^{N_sl_s}}(a_i+j p^t)
=T_{n^{N_sl_s}}(a_i)+\sum^{n^{N_sl_s}}_{m=1}(jp^{t-1})^m\frac{T^{(m)}_{n^{N_sl_s}}(a_i) p^m}{m!},
\]
where $j\in\mathbb{Z}$.
Substituting congruences~\eqref{eq:TnXk:pw1} and \eqref{eq:TnXk:pw2} into this equation, and noting that $m\cdot (t-1)+w+2\ge w+t+1$ for $m\ge 2$, one obtains
\begin{equation}\label{eq:tnlx}
T_{n^{N_sl_s}}(a_i+j p^t)\equiv T_{n^{N_sl_s}}(a_i)+j p^t T'_{n^{N_sl_s}}(c_i)\pmod{p^{w+t+1}}.
\end{equation}

If $l_s=1$, i.e., $T'_{n^{N_s}}(s)\equiv 1\pmod p$, substituting $n$ in Proposition~\ref{pro:Tnx=x:1} with $n^{N_s}$, there exists an element in $\mathcal{Z}_{p^k}$ that satisfies condition~\eqref{condition}.
If $l_s\neq 1$, substituting $n$ in Proposition~\ref{pro:Tnx=x:2} with $n^{N_s}$, one finds a unique element $s_k\in\mathcal{Z}_{p^k}$ that satisfies
condition~\eqref{condition}.
In both cases, the notation $s_k$ refers to the element satisfying condition~\eqref{condition}, which generates a single cycle of length $N_s$ by self-iteration in $G_{p^k}$. Denote this cycle as $\{c'_i\}_{i=0}^{N_s-1}$, where $c'_i=T^i_n(s_k)\bmod p^k$.
These states in the cycle satisfy condition~\eqref{condition}, namely
$c'_i \equiv c_i\pmod p$ and
\begin{equation}\label{eq:Tnc_i:c_i}
T_{n^{N_s}}(c'_i)\equiv c'_i\pmod{p^k}.
\end{equation}

Next, the analysis proceeds in two cases:
\begin{itemize}[\newItemizeWidth]
\item $k\leq w+1$: Applying congruence~\eqref{eq:tnlx} with $(a_i, t)=(c'_i, 1)$, 
one obtains $T_{n^{N_sl_s}}(c'_i+jp)\equiv T_{n^{N_sl_s}}(c'_i)+jpT'_{n^{N_sl_s}}(c_i)\pmod{p^{w+2}}$.
Adopting congruence~\eqref{eq:Tnc_i:c_i} and that in 
condition~\eqref{eq:tn'x:change} into the preceding congruence, one deduces
\[
T_{n^{N_sl_s}}(c'_i+jp) \equiv c'_i+jp\pmod{p^k}.
\]
Thus, the least period of the sequence $\{T_n^i(a)\bmod p^k\}_{i\ge 0}$ is $N_sl_s$ when $a\equiv c'_i\pmod p$ and $a\neq c'_i$.
Excluding the cycle $\{c'_i\}_{i=0}^{N_s-1}$,
the elements in the set
$\bigcup_{i=0}^{N_s-1}\{x \vvert x\in \mathcal{Z}_{p^k}, 
x\equiv c'_i \pmod p\}$
form $\frac{N_sp^{k-1}-N_s}{N_sl_s}=\frac{p^{k-1}-1}{l_s}$ cycles of length $N_sl_s$.

\item $k>w+1$: Applying congruence~\eqref{eq:tnlx} with $(a_i, t)=(c'_i, k-w)$, one deduces $T_{n^{N_sl_s}}(c'_i+jp^{k-w})\equiv T_{n^{N_sl_s}}(c'_i)+j p^{k-w}T'_{n^{N_sl_s}}(c_i)\pmod{p^{k+1}}$.
Using condition~\eqref{eq:tn'x:change} and congruence~\eqref{eq:Tnc_i:c_i} as above, one
concludes
\[
T_{n^{N_sl_s}}(c'_i+j p^{k-w}) \equiv c'_i+j p^{k-w}\pmod{p^k}.
\]
Therefore, the elements in set $\bigcup_{i=0}^{N_s-1}\{x \vvert x\in\mathcal{Z}_{p^k}, x\equiv c'_i\pmod{p^{k-w}}\}$ 
form a single cycle of length $N_s$ and 
$\frac{N_sp^w-N_s}{N_sl_s}=\frac{p^w-1}{l_s}$ cycles of length $N_sl_s$.
Furthermore, one considers the cases $t\in\{1, 2, \cdots, k-1-w\}$.
Since $w+t+1\le k$, one has $T_{n^{N_sl_s}}(c'_i)\equiv c'_i\pmod{p^{w+t+1}}$ from congruence~\eqref{eq:Tnc_i:c_i}.
Incorporating this congruence and condition~\eqref{eq:tn'x:change} into congruence~\eqref{eq:tnlx} with $(a_i, j)=(c'_i, j')$, one obtains
\begin{equation}\label{eq:Tnxk:temp3}
    \begin{cases}
    T_{n^{N_sl_s}}(c'_i+j' p^t)\equiv c'_i+j' p^t\pmod{p^{w+t}};\\
    T_{n^{N_sl_s}}(c'_i+j' p^t)\not\equiv c'_i+j' p^t\pmod{p^{w+t+1}}, 
    \end{cases}
\end{equation}
where $j'\in\mathbb{Z}$ and $p\nmid j'$.
Referring to Theorem~\ref{Theorem:periodValue} with $v=w+t$, one has the least period of the sequence $\{T_n^i(a)\bmod p^k\}_{i\ge 0}$ is $N_sl_sp^{k-(w+t)}$ for any $a\in\bigcup_{i=0}^{N_s-1}\{x \vvert x\in \mathcal{Z}_{p^k}, x\equiv c'_i\pmod{p^t},
x\not\equiv c'_i\pmod{p^{t+1}}\}$.
Thus, the elements in this set form
\(
\frac{N_s\cdot (p^{k-t}-p^{k-t-1})}{N_sl_s\cdot p^{k-w-t}}=\frac{(p-1)p^{w-1}}{l_s}
\)
cycles of length \(N_sl_sp^{k-w-t}\).
\end{itemize}
Besides such cycles of length greater than $N_sl_s$, which are available only if $k-1-w>0$, there is one cycle of length $N_s$ and
$\frac{p^{\min(k-1, w)}-1}{l_s}$ cycles of length $N_sl_s$ in either case discussed above.
\end{proof}

Using $G_{p^k}$ with the parameters $(n, p) = (13, 7)$, as depicted in Fig.~\ref{fig:SMNcheby:n13p7}, one may further illustrate and verify Proposition~\ref{pro:SMN:+1}. By calculation, one finds $w = \max\{e \vvert n^{2 \cdot \ord(n^2)} \equiv 1 \pmod{p^e}\} = 1$. For $s = 2$, one computes $N_s = 2$, $l_s = 1$, and the set $\{T^i_{13}(2) \bmod 7\}_{i=0}^{N_s-1} = \{2, 5\}$. When $k = 1$, one observes an elementary cycle ``$2 \rightarrow 5 \rightarrow 2$" in Fig.~\ref{fig:SMNcheby:n13p7}a), yielding $c_0 = 2$ and $c_1 = 5$.
Substituting $n = 13^2$ into Proposition~\ref{pro:Tnx=x:1}, one finds $z_s = 1$. Furthermore, $s_2 = 2$ and $s_3 = 44$ satisfy condition~\eqref{condition}. 
When $k = 2$, one obtains $c'_0 = 2$ and $c'_1 = T_{13}(2) \bmod 7^2 = 12$. The elements in the set $\bigcup_{i=0}^{1} \{x \vvert x\in \mathcal{Z}_{p^2}, x \equiv c'_i \pmod{7}\}$ form an elementary cycle ``$2 \rightarrow 12 \rightarrow 2$" along with $\frac{2 \cdot (7 - 1)}{2} = 6$ additional cycles of the same length. These cycles are shown within the dashed rectangle in Fig.~\ref{fig:SMNcheby:n13p7}b).
When $k = 3$, one finds $c'_0 = 44$ and $c'_1 = T_{13}(44) \bmod 7^3 = 152$. The elements in the set $\cup_{i=0}^{1} \{x \vvert x\in \mathcal{Z}_{p^3}, x \equiv c'_i \pmod{7^2}\}$ form an elementary cycle ``$44 \rightarrow 152 \rightarrow 44$" along with $\frac{2 \cdot (7 - 1)}{2} = 6$ additional cycles of equal length. These are depicted within the dashed rectangle in Fig.~\ref{fig:SMNcheby:n13p7}c).
Next, one considers the case where the modulo in condition~\eqref{eq:Tnxk:temp3} is taken with a lower power of $p$. The elements in the set $\cup_{i=0}^{1} \{x \vvert x\in \mathcal{Z}_{7^3}, x \equiv c'_i \pmod{7}, x \not\equiv c'_i \pmod{7^2}\}$ yield $\frac{2\cdot (7-1)\cdot 7}{2 \cdot 1 \cdot 7} = 6$ cycles of length $2 \cdot 1 \cdot 7 = 14$. An example of such a cycle is ``$2 \rightarrow 208 \rightarrow 149 \rightarrow 12 \rightarrow 296 \rightarrow 159 \rightarrow 100 \rightarrow 306 \rightarrow 247 \rightarrow 110 \rightarrow 51 \rightarrow 257 \rightarrow 198 \rightarrow 61 \rightarrow 2$". These cycles are presented within the solid rectangle in Fig.~\ref{fig:SMNcheby:n13p7}c).
As $k$ increases, one observes that the number of cycles of any given length stabilizes, yielding a complete description of the whole functional graph's structure, as described in Theorem~\ref{Theo:wTc}.

\begin{theorem}\label{Theo:wTc}
The number of cycles of length $T_c$ in the functional graph of the Chebyshev permutation polynomial $T_n(x)$ over $\Zp{k}$, $C_{T_c, k}$, is invariant for $k>w+1$, i.e.,
\begin{equation*}
C_{T_c, k}=C_{T_c, w+1},
\end{equation*}
where $w=\max\{e \vvert n^{2\cdot\ord(n^2)} \equiv 1\pmod{p^e}\}$.
\end{theorem}
\begin{proof}
Consider an integer $s\in\mathcal{Z}_p$. Referring to Proposition~\ref{pro:SMN:+1}, all $p^{k-1}N_s$ elements in the set $\{x \vvert 
x\in \mathcal{Z}_{p^k}, (x\bmod p)\in\{T_n^i(s)\bmod p\}_{i=0}^{N_s-1}\}$ form one cycle of length $N_s$, $\frac{p^{w}-1}{l_s}$ cycles of length $N_sl_s$, and $\frac{(p-1)p^{w-1}}{l_s}$ cycles of length $N_sl_sp^{k-w-t}$ when $k>w+1$, where
$t\in\{1, 2, \cdots, k-w-1\}$.
Therefore, the number of cycles of any length depends only on the parameters $n$, $p$, and $l_s$, independent of $k$.
One proves the theorem by iteratively traversing the elements
in $\mathcal{Z}_p$ that have not yet been covered by the preceding set.
\end{proof}

\setlength\tabcolsep{7pt} 
\begin{table}[!htb]
 \centering
 \caption{The number of cycles of period $T_c$ in $G_{p^k}$ when $(n, p)=(11, 3)$.}
    \begin{tabular}{*{7}{c|}c}
    \hline
    \diagbox{$k$}{$C_{T_c, k}$}{$T_c$}
    &1	&3	&9	&27	&81	&243	&729 \\ \hline
    1	&3	&	&	&	&	&	&    \\ \hline
    2	&9	&	&	&	&	&	&    \\ \hline
    3	&15	&4	&	&  &	&	&    \\ \hline
    4	&15	&10	&4	&	&	&	&    \\ \hline
    5	&15	&10	&10	&4	&	&	&    \\ \hline
    6	&15	&10	&10	&10	&4	&	&    \\ \hline
    7	&15	&10	&10	&10	&10	&4	&    \\ \hline
    8	&15	&10	&10	&10	&10	&10	&4   \\ \hline
    \end{tabular}
\label{table:cycle:n11p3}
\end{table}

Theorem~\ref{Theo:wTc} is verified by the cycle distribution within the functional graph $G_{p^k}$ for various sets of $(n, p, k)$. 
The four sets of \((n, p)\) used in Fig.~\ref{fig:cycle:disall} are \((13, 7)\), \((43, 5)\), \((19, 7)\), and \((443, 5)\).
The corresponding thresholds are $2, 3, 4$, and $5$ by
calculating from Theorem~\ref{Theo:wTc}.
The results are consistent with the corresponding sub-figure in Fig.~\ref{fig:cycle:disall}. However, it can be verified that Lemma~\ref{lemma:pw:Tn'1} does not hold for 
$p=3$. This causes the evolution rule of the Chebyshev permutation polynomial over \(G_{3^k}\) not to satisfy Theorem~\ref{Theo:wTc}, as verified by Table~\ref{table:cycle:n11p3}.

\section{Application of the found properties of $G_{p^k}$}
\label{sec:application}

This section introduces a strategy for selecting parameters and initial states of Chebyshev polynomials to avoid potential vulnerabilities in key exchange protocols based on these polynomials, as designed in \cite{Yoshioka:ChebyshevPk:TCAS2:2018}.
The same strategy can be applied to public-key encryption algorithms that utilize Chebyshev polynomials over $\Zp{k}$.

Consider a scenario in which Alice and Bob wish to agree on a secret key within $\mathcal{Z}_{p^k}$. Initially, one selects an integer $n$ that satisfies condition~\eqref{PermuteCondition}.
Note that the parameters $n$, $p$, $k$, and the initial state $a \in \mathcal{Z}_{p^k}$ are public. Alice begins by transmitting $T_{n^{N_1}}(a)\bmod p^k$ to Bob. Bob, in response, sends $T_{n^{N_2}}(a) \bmod p^k$ back to Alice. At this point, both parties can compute the shared secret key 
$\lambda=T_{n^{N_2}}(T_{n^{N_1}}(a))\bmod p^k=T_{n^{N_1}}(T_{n^{N_2}}(a))\bmod p^k=T_{n^{N_1+N_2}}(a)\bmod p^k$. Note that parameters $N_1$ and $N_2$ are kept secret by Alice and Bob, respectively. As previously discussed, the number of valid values for $\lambda$ is determined by the least period of the Chebyshev integer sequence $\{T_n^i(a) \bmod p^k\}_{i\ge 0}$. A poorly chosen set of parameters or an unfortunate initial state can drastically reduce the size of the \emph{session key space}, thereby increasing the likelihood of a successful attack.
Figures~\ref{fig:SMNcheby:n13p7} and \ref{fig:cycle:disall} illustrate that certain cycles exhibit very short periods, regardless of $k$. However, by carefully leveraging the properties of the functional graph of Chebyshev polynomials, one can develop a method to select parameters and initial states that avoid such vulnerabilities.

Referencing Proposition~\ref{pro:SMN:+1}, one finds that smaller values of $w$ are preferred to produce cycles of longer length within the domain $\mathcal{Z}_{p^k}$. Specifically, when $w=t=1$, the number of the cycles whose length is greater than $N_sl_s$ is given by
$\frac{(p-1)p^{w-1}}{l_s}=\frac{p-1}{l_s}$
and the cycle length becomes $N_sl_sp^{k-w-t}=N_sl_sp^{k-2}$. 
To achieve this, one selects $n$ such that \( w=\max\{e \vvert n^{2\cdot\ord(n^2)} \equiv 1\pmod{p^e}\}=1\).
Recalling the discussion on the case $k>w+1$ in the proof of Proposition~\ref{pro:SMN:+1}, 
one chooses an initial state $a\in \bigcup_{i=0}^{N_s-1} \{x \vvert x\in \mathcal{Z}_{p^k}, x\equiv T^i_n(s)\pmod{p}, x\not\equiv T^i_n(s_{2})\pmod{p^2}\}$ to ensure $t=1$, where $s_{2}$ is iteratively computed according to 
Proposition~\ref{pro:Tnx=x:1} or \ref{pro:Tnx=x:2} depending on 
whether $T'_{n^{N_s}}(s)\equiv 1\pmod p$, with $s$ being a predefined element in $\mathcal{Z}_p$ as used in Sec.~\ref{sec:network}.

For example, when $(n, p, k, s)=(11, 7, 8, 3)$, one calculates $\max\{e \vvert n^{2\cdot\ord(n^2)} \equiv 1\pmod{p^e}\}=1$, $N_s=1$ and $l_s=6$. According to Proposition \ref{pro:Tnx=x:2}, one determines $s_2=24$ and $s_8=2882400$.
To optimize the session key space, one selects $a$ such that $a \equiv 3\pmod{7}$ and $a \not\equiv 24\pmod{7^2}$, which implies $a=3 + 7j$, where $j \not\equiv 3\pmod{7}$. This results in a session key space of size $1 \cdot 6 \cdot 7^6=705894$.
Conversely, if $a\equiv 2882400\pmod{7^6}$, the maximal size of the corresponding session key space is only $1\cdot 6 \cdot 7=42$.

As discussed at the beginning of Sec.~\ref{subsec:cheby:integer}, understanding the properties of Chebyshev polynomials aids in defining the parameter $l_x$ appropriately. The findings presented in this paper may be used in recent cryptographic applications of Chebyshev polynomials in \cite{chatterjee2018DSC, Dariush:Chebyshevprivacy:TII2020, Lee:cheby-approxi:TDSC22, chen:PRNS:TC22, Tomar:auth:TMC2024} from both cryptographic and cryptanalytic perspectives.

\section{Conclusion}

This paper investigated the structure of the functional graph 
constituted by a Chebyshev permutation polynomial over $\Zp{k}$, where $p$ is a prime number greater than three.
The complete explicit expression of the least period of the sequence
generated by iterating the polynomial from any given initial state was presented.
Combining the properties of
Chebyshev polynomials and their derivatives, the paper discloses
how the graph structure evolves concerning parameter $k$.
It can be concluded that the graph structure of the Chebyshev permutation polynomial is quite regular, and 
a certain number of short cycles exist regardless of the parameter's value.
The work in this paper can serve as a solid basis for studying the graph structure of Chebyshev polynomials over various algebraic domains, such as ring $\mathbb{Z}_q$, 
where $q$ is a composite number.

\bibliographystyle{IEEEtran_doi}
\bibliography{V3_Double-Chebyshev-0918}

\begin{thebibliography}{10}
\providecommand{\url}[1]{#1}
\csname url@samestyle\endcsname
\providecommand{\newblock}{\relax}
\providecommand{\bibinfo}[2]{#2}
\providecommand{\BIBentrySTDinterwordspacing}{\spaceskip=0pt\relax}
\providecommand{\BIBentryALTinterwordstretchfactor}{4}
\providecommand{\BIBentryALTinterwordspacing}{\spaceskip=\fontdimen2\font plus
\BIBentryALTinterwordstretchfactor\fontdimen3\font minus
  \fontdimen4\font\relax}
\providecommand{\BIBforeignlanguage}[2]{{%
\expandafter\ifx\csname l@#1\endcsname\relax
\typeout{** WARNING: IEEEtran.bst: No hyphenation pattern has been}%
\typeout{** loaded for the language `#1'. Using the pattern for}%
\typeout{** the default language instead.}%
\else
\language=\csname l@#1\endcsname
\fi
#2}}
\providecommand{\BIBdecl}{\relax}
\BIBdecl

\bibitem{Chebyshev:life:JAT99}
P.~Butzer and F.~Jongmans, ``\href{http://dx.doi.org/10.1006/jath.1998.3289}{P.
  l. {C}hebyshev (1821--1894): a guide to his life and work},'' \emph{Journal
  of Approximation Theory}, vol.~96, no.~1, pp. 111--138, 1999.

\bibitem{Lee:cheby-approxi:TDSC22}
E.~Lee, J.-W. Lee, J.-S. No, and Y.-S. Kim,
  ``\href{http://dx.doi.org/10.1109/TDSC.2021.3105111}{Minimax approximation of
  sign function by composite polynomial for homomorphic comparison},''
  \emph{IEEE Transactions on Dependable and Secure Computing}, vol.~19, no.~6,
  pp. 3711--3727, 2022.

\bibitem{Liu:Chebyrandom:2011}
N.-S. Liu,
  ``\href{http://dx.doi.org/10.1016/j.cnsns.2010.04.021}{Pseudo-randomness and
  complexity of binary sequences generated by the chaotic system},''
  \emph{Communications in Nonlinear Science and Numerical Simulation}, vol.~16,
  no.~2, pp. 761--768, 2011.

\bibitem{Chencc:chebyshev:TCASI2001}
C.-C. Chen, K.~Yao, K.~Umeno, and E.~Biglieri,
  ``\href{http://dx.doi.org/10.1109/81.948438}{Design of spread-spectrum
  sequences using chaotic dynamical systems and ergodic theory},'' \emph{IEEE
  Transactions on Circuits and Systems I: Fundamental Theory and Applications},
  vol.~48, no.~9, pp. 1110--1114, 2001.

\bibitem{chatterjee2018DSC}
S.~Chatterjee, S.~Roy, A.~K. Das, S.~Chattopadhyay, N.~Kumar, and A.~V.
  Vasilakos, ``\href{http://dx.doi.org/10.1109/TDSC.2016.2616876}{Secure
  biometric-based authentication scheme using {C}hebyshev chaotic map for
  multi-server environment},'' \emph{IEEE Transactions on Dependable and Secure
  Computing}, vol.~15, no.~5, pp. 824--839, 2018.

\bibitem{Tomar:auth:TMC2024}
A.~Tomar and S.~Tripathi, ``\href{http://dx.doi.org/10.1109/TMC.2024.3357599}{A
  {C}hebyshev polynomial-based authentication scheme using blockchain
  technology for {Fog}-based vehicular network},'' \emph{IEEE Transactions on
  Mobile Computing}, vol.~23, no.~9, 2024.

\bibitem{Dariush:Chebyshevprivacy:TII2020}
D.~{Abbasinezhad-Mood}, A.~{Ostad-Sharif}, S.~M. {Mazinani}, and
  M.~{Nikooghadam},
  ``\href{http://dx.doi.org/10.1109/TII.2020.2974258}{Provably secure
  escrow-less {C}hebyshev chaotic map-based key agreement protocol for vehicle
  to grid connections with privacy protection},'' \emph{IEEE Transactions on
  Industrial Informatics}, vol.~16, no.~12, pp. 7287--7294, 2020.

\bibitem{Zhang:Privacy:TII:2022}
L.~Zhang, W.~Gao, S.~Chen, W.~Ren, K.-K.~R. Choo, and N.~N. Xiong,
  ``\href{http://dx.doi.org/10.1109/TII.2021.3133566}{A privacy-preserving
  proximity testing using private set intersection for vehicular ad-hoc
  networks},'' \emph{IEEE Transactions on Industrial Informatics}, vol.~18,
  no.~10, pp. 7373--7383, 2022.

\bibitem{Rivlin:chebyshev:1990}
T.~J. Rivlin, \emph{Chebyshev {P}olynomials: {F}rom {A}pproximation {T}heory to
  {A}lgebra and {N}umber {T}heory}.\hskip 1em plus 0.5em minus 0.4em\relax
  Wiley, 1990.

\bibitem{Bergamo:PKIchebyshev:CASI2005}
P.~Bergamo, P.~D'Arco, A.~D. Santis, and L.~Kocarev,
  ``\href{http://dx.doi.org/10.1109/TCSI.2005.851701}{Security of public-key
  cryptosystems based on {C}hebyshev polynomials},'' \emph{IEEE Transactions on
  Circuits and Systems I--Regular Papers}, vol.~52, no.~7, pp. 1382--1393,
  2005.

\bibitem{Liao:ITC:2010}
X.~Liao, F.~Chen, and K.-W. Wong,
  ``\href{http://dx.doi.org/10.1109/TC.2010.148}{On the security of public-key
  algorithms based on {C}hebyshev polynomials over the finite field
  $\mathbb{Z}_n$},'' \emph{IEEE Transactions on Computers}, vol.~59, no.~10,
  pp. 1392--1401, 2010.

\bibitem{Kocarev:2003}
L.~{Kocarev} and Z.~{Tasev},
  ``\href{http://dx.doi.org/10.1109/ISCAS.2003.1204947}{Public-key encryption
  based on {C}hebyshev maps},'' in \emph{Proceedings of International Symposium
  on Circuits and Systems}, vol.~3, 2003, pp. 28--31.

\bibitem{Cheong:PKIchebyshev:CASII2007}
K.~Y. Cheong and T.~Koshiba,
  ``\href{http://dx.doi.org/10.1109/TCSII.2007.900875}{More on security of
  public-key cryptosystems based on {C}hebyshev polynomials},'' \emph{IEEE
  Transactions on Circuits and Systems II-Express Briefs}, vol.~54, no.~9, pp.
  795--799, 2007.

\bibitem{Kocarev:2005}
L.~Kocarev, J.~Makraduli, and P.~Amato,
  ``\href{http://dx.doi.org/10.1007/s00034-005-2403-x}{Public-key encryption
  based on {C}hebyshev polynomials},'' \emph{Circuits Systems and Signal
  Processing}, vol.~24, no.~5, pp. 497--517, 2005.

\bibitem{Chen:chebyZN:IS2011}
F.~Chen, X.~Liao, T.~Xiang, and H.~Zheng,
  ``\href{http://dx.doi.org/10.1016/j.ins.2011.07.008}{Security analysis of the
  public key algorithm based on {C}hebyshev polynomials over the integer ring
  $\mathbb{Z}_n$},'' \emph{Information Sciences}, vol. 181, no.~22, pp.
  5110--5118, 2011.

\bibitem{Yoshioka:ChebyshevPk:TCAS2:2020}
D.~{Yoshioka}, ``\href{http://dx.doi.org/10.1109/TCSII.2019.2954855}{Security
  of public-key cryptosystems based on {C}hebyshev polynomials over
  $\mathbb{Z}/p^{k}\mathbb{Z}$},'' \emph{IEEE Transactions on Circuits and
  Systems II: Express Briefs}, vol.~67, no.~10, pp. 2204--2208, 2020.

\bibitem{Yoshioka:Chebyshev2k:TCAS2:2016}
D.~Yoshioka and Y.~Dainobu,
  ``\href{http://dx.doi.org/10.1109/TCSII.2016.2531058}{Periodic properties of
  {C}hebyshev polynomial sequences over the residue ring
  $\mathbb{Z}/2^{k}\mathbb{Z}$},'' \emph{IEEE Transactions on Circuits and
  Systems II: Express Briefs}, vol.~63, no.~8, pp. 778--782, 2016.

\bibitem{Yoshioka:ChebyshevPk:TCAS2:2018}
D.~Yoshioka, ``\href{http://dx.doi.org/10.1109/TCSII.2017.2739190}{Properties
  of {C}hebyshev polynomials modulo $p^k$},'' \emph{IEEE Transactions on
  Circuits and Systems II: Express Briefs}, vol.~65, no.~3, pp. 386--390, 2018.

\bibitem{Umeno:permute:2018}
A.~Iwasaki and K.~Umeno,
  ``\href{http://dx.doi.org/10.1007/s13160-017-0275-7}{Three theorems on odd
  degree {C}hebyshev polynomials and more generalized permutation polynomials
  over a ring of module $2^w$},'' \emph{Japan Journal of Industrial and Applied
  Mathematics}, vol.~35, pp. 49--69, 2018.

\bibitem{Rudolf:Introfinite:1994}
R.~Lidl and H.~Niederreiter,
  \emph{\href{http://dx.doi.org/10.1017/CBO9781139172769}{Introduction to
  finite fields and their applications}}, 2nd~ed.\hskip 1em plus 0.5em minus
  0.4em\relax Cambridge, UK: Cambridge University Press, 1994.

\bibitem{chen:PRNS:TC22}
S.~Chen, S.~Ma, Z.~Qin, B.~Zhu, Z.~Xiao, and M.~Liu,
  ``\href{http://dx.doi.org/10.1109/TC.2022.3143702}{A low complexity and long
  period digital random sequence generator based on residue number system and
  permutation polynomial},'' \emph{IEEE Transactions on Computers}, vol.~71,
  no.~11, pp. 3008--3017, 2022.

\bibitem{cqli:network:TCASI2019}
C.~Li, B.~Feng, S.~Li, J.~Kurths, and G.~Chen,
  ``\href{http://dx.doi.org/10.1109/TCSI.2018.2888688}{Dynamic analysis of
  digital chaotic maps via state-mapping networks},'' \emph{IEEE Transactions
  on Circuits and Systems I: Regular Papers}, vol.~66, no.~6, pp. 2322--2335,
  2019.

\bibitem{cqli:Cat:TC22}
C.~Li, K.~Tan, B.~Feng, and J.~L\"u,
  ``\href{http://dx.doi.org/10.1109/TC.2021.3051387}{The graph structure of the
  generalized discrete {A}rnold's {C}at map},'' \emph{IEEE Transactions on
  Computers}, vol.~71, no.~2, pp. 364--377, 2022.

\bibitem{licq:Logistic:IJBC2023}
X.~Lu, E.~Y. Xie, and C.~Li,
  ``\href{http://dx.doi.org/10.1142/S0218127423500633}{Periodicity analysis of
  {L}ogistic map over ring $\mathbb{Z}_{3^n}$},'' \emph{International Journal
  of Bifurcation and Chaos}, vol.~33, no.~5, p. art. no. 2350063, 2023.

\bibitem{Chen:cat:TIT2012}
F.~Chen, K.-W. Wong, X.~Liao, and T.~Xiang,
  ``\href{http://dx.doi.org/10.1109/TIT.2011.2171534}{Period distribution of
  generalized discrete {A}rnold {C}at map for $n=p^{e}$},'' \emph{IEEE
  Transactions on Information Theory}, vol.~58, no.~1, pp. 445--452, 2012.

\bibitem{chenf:cat2:TIT13}
------, ``\href{http://dx.doi.org/10.1109/TIT.2012.2235907}{Period distribution
  of the generalized discrete {A}rnold {C}at map for $n=2^e$},'' \emph{IEEE
  Transactions on Information Theory}, vol.~59, no.~5, pp. 3249--3255, 2013.

\bibitem{Li:ISIT2024}
K.~Tan and C.~Li,
  ``\href{http://dx.doi.org/10.1109/ISIT57864.2024.10619562}{Network analysis
  of {B}aker's map implemented in a fixed-point arithmetic domain},'' in
  \emph{Proceedings of IEEE International Symposium on Information Theory (ISIT
  2024)}.\hskip 1em plus 0.5em minus 0.4em\relax IEEE, 2024, pp. 1332--1337.

\bibitem{Chang:cycleLFSR:CC18}
Z.~Chang, M.~F. Ezerman, S.~Ling, and H.~Wang,
  ``\href{http://dx.doi.org/10.1007/s12095-017-0273-2}{The cycle structure of
  {LFSR} with arbitrary characteristic polynomial over finite fields},''
  \emph{Cryptography and Communications}, vol.~10, no.~6, pp. 1183--1202, 2018.

\bibitem{gongg:cycle:TIT19}
Z.~Chang, G.~Gong, and Q.~Wang,
  ``\href{http://dx.doi.org/10.1109/TIT.2019.2956741}{Cycle structures of a
  class of cascaded {FSRs}},'' \emph{IEEE Transactions on Information Theory},
  vol.~66, no.~6, pp. 3766--3774, 2020.

\bibitem{Daniel:Redei:DM2015}
C.~Qureshi and D.~Panario, ``\href{http://dx.doi.org/Qureshi, Claudio and
  Panario, Daniel}{R{\'e}dei actions on finite fields and multiplication map in
  cyclic group},'' \emph{SIAM Journal on Discrete Mathematics}, vol.~29, pp.
  1486--1503, 2015.

\bibitem{Daniel:Chebyshev:2018}
------, ``\href{http://dx.doi.org/10.1007/s10623-018-0545-7}{The graph
  structure of {C}hebyshev polynomials over finite fields and applications},''
  \emph{Designs, Codes and Cryptography}, vol.~87, pp. 393--416, 2019.

\bibitem{Daniel:lineargraph:DCC2019}
D.~Panario and L.~Reis,
  ``\href{http://dx.doi.org/10.1007/s10623-018-0547-5}{The functional graph of
  linear maps over finite fields and applications},'' \emph{Designs, Codes and
  Cryptography}, vol.~87, no. 2-3, pp. 437--453, 2019.

\bibitem{Li:neteork:ISCAS2019}
C.~Li, J.~Lu, and G.~Chen,
  ``\href{http://dx.doi.org/10.1109/ISCAS.2019.8702232}{Network analysis of
  chaotic dynamics in fixed-precision digital domain},'' in \emph{2019 IEEE
  International Symposium on Circuits and Systems (ISCAS)}, 2019, pp. 1--5.

\bibitem{gassert:chebyshevaction:DM2014}
T.~A. Gassert,
  ``\href{http://dx.doi.org/10.1016/j.disc.2013.10.014}{{C}hebyshev action on
  finite fields},'' \emph{Discrete Mathematics}, vol. 315, pp. 83--94, 2014.

\bibitem{Yoshioka:ISIT2023}
D.~Yoshioka,
  ``\href{http://dx.doi.org/10.1109/ISIT54713.2023.10206860}{Periodic
  properties of commutative polynomials defined by fourth order recurrence
  relations with two variables over $\mathbb{Z}_{2^k}$},'' in \emph{Proceedings
  of IEEE International Symposium on Information Theory (ISIT 2023)}.\hskip 1em
  plus 0.5em minus 0.4em\relax IEEE, 2023, pp. 1425--1429.

\bibitem{Weng:notePP:TIT2008}
G.~Weng and C.~Dong, ``\href{http://dx.doi.org/10.1109/TIT.2008.926426}{A note
  on permutation polynomials over $\mathbb{Z}_{n}$},'' \emph{IEEE Transactions
  on Information Theory}, vol.~54, no.~9, pp. 4388--4390, 2008.

\bibitem{Sun:newCPP:TIT2021}
B.~Sun, K.~Li, J.~Guo, and L.~Qu,
  ``\href{http://dx.doi.org/10.1109/TIT.2021.3100756}{New constructions of
  complete permutations},'' \emph{IEEE Transactions on Information Theory},
  vol.~67, no.~11, pp. 7561--7567, 2021.

\bibitem{Wu:CPP:TIT2023}
G.~Wu, K.~Feng, N.~Li, and T.~Helleseth,
  ``\href{http://dx.doi.org/10.1109/TIT.2023.3238994}{New results on the $-$1
  conjecture on cross-correlation of $m$-sequences based on complete
  permutation polynomials},'' \emph{IEEE Transactions on Information Theory},
  vol.~69, no.~6, pp. 4035--4044, 2023.

\bibitem{Lidl:Dick:1993}
R.~Lidl, G.~L. Mullen, and G.~Turnwald, \emph{Dickson Polynomials}.\hskip 1em
  plus 0.5em minus 0.4em\relax Chapman and Hall/CRC, 1993.

\bibitem{Moll:numbersandfunctions:2012}
V.~H. Moll, \emph{Numbers and Functions: From a Classical-Experimental
  Mathematician's Point of View}.\hskip 1em plus 0.5em minus 0.4em\relax
  American Mathematical Society, 2012.

\bibitem{Mihet:2010:Kummer}
D.~Mihet, ``\href{http://dx.doi.org/10.1007/s12045-010-0123-4}{Legendre’s and
  {K}ummer’s theorems again},'' \emph{Resonance}, vol.~15, pp. 1111--1121,
  2010.

\bibitem{Wan:Lecture:2003}
Z.-X. Wan, \emph{\href{http://dx.doi.org/10.1142/5350}{Lectures on finite
  fields and Galois rings}}.\hskip 1em plus 0.5em minus 0.4em\relax World
  Scientific, 2003.

\bibitem{2019NotesOE}
F.~Qi, D.-W. Niu, and D.~Lim,
  ``\href{http://dx.doi.org/10.18514/MMN.2019.2976}{Notes on explicit and
  inversion formulas for the {C}hebyshev polynomials of the first two kinds},''
  \emph{Miskolc Mathematical Notes}, vol.~20, no.~2, pp. 1129--1137, 2019.

\bibitem{Ap76:analyticnumbertheory}
T.~M. Apostol,
  \emph{\href{http://dx.doi.org/10.1007/978-1-4757-5579-4}{Introduction to
  Analytic Number Theory}}, 1st~ed.\hskip 1em plus 0.5em minus 0.4em\relax New
  York: Springer-Verlag, 1976.

\bibitem{Rayes:Factorization:CM05}
M.~O. Rayes, V.~Trevisan, and P.~S. Wang,
  ``\href{http://dx.doi.org/10.1016%2Fj.camwa.2005.07.003}{Factorization
  properties of {C}hebyshev polynomials},'' \emph{Computers and Mathematics
  with Applications}, vol.~50, no. 8-9, p. 1231–1240, 2005.

\end{thebibliography}

\vskip -1\baselineskip plus -1fil
\graphicspath{{author_figures_pdf/}}
\begin{IEEEbiography}[{\includegraphics[width=1.1in, height=1.25in,clip,keepaspectratio]{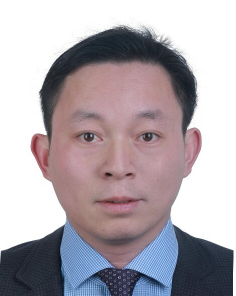}}]{Chengqing Li}(M'07-SM'13)
received his M.Sc. degree in Applied Mathematics from Zhejiang University, China in 2005 and
his Ph.D. degree in Electronic Engineering from City University of Hong Kong in 2008. Thereafter,
he worked as a Postdoctoral Fellow at The Hong Kong Polytechnic University till September 2010.
Then, he worked at the College of Information Engineering, Xiangtan University, China. From April 2013 to July 2014, he worked at the
University of Konstanz, Germany, under the support of the Alexander von Humboldt Foundation.
From May 2020, he have been working at School of Computer Science, Xiangtan University, China as the dean.
He is an associate editor for the International Journal of Bifurcation and Chaos and Signal Processing.

Prof. Li focuses on the security analysis of multimedia encryption and privacy protection schemes.
He has published over sixty papers on the subject in the past 20 years, receiving more than 6000 citations with an h-index of 41.
He is a Fellow of IET.
\end{IEEEbiography}
\vskip -2\baselineskip plus -1fil

\begin{IEEEbiography}[{\includegraphics[width=1in,height=1.25in,clip,keepaspectratio]{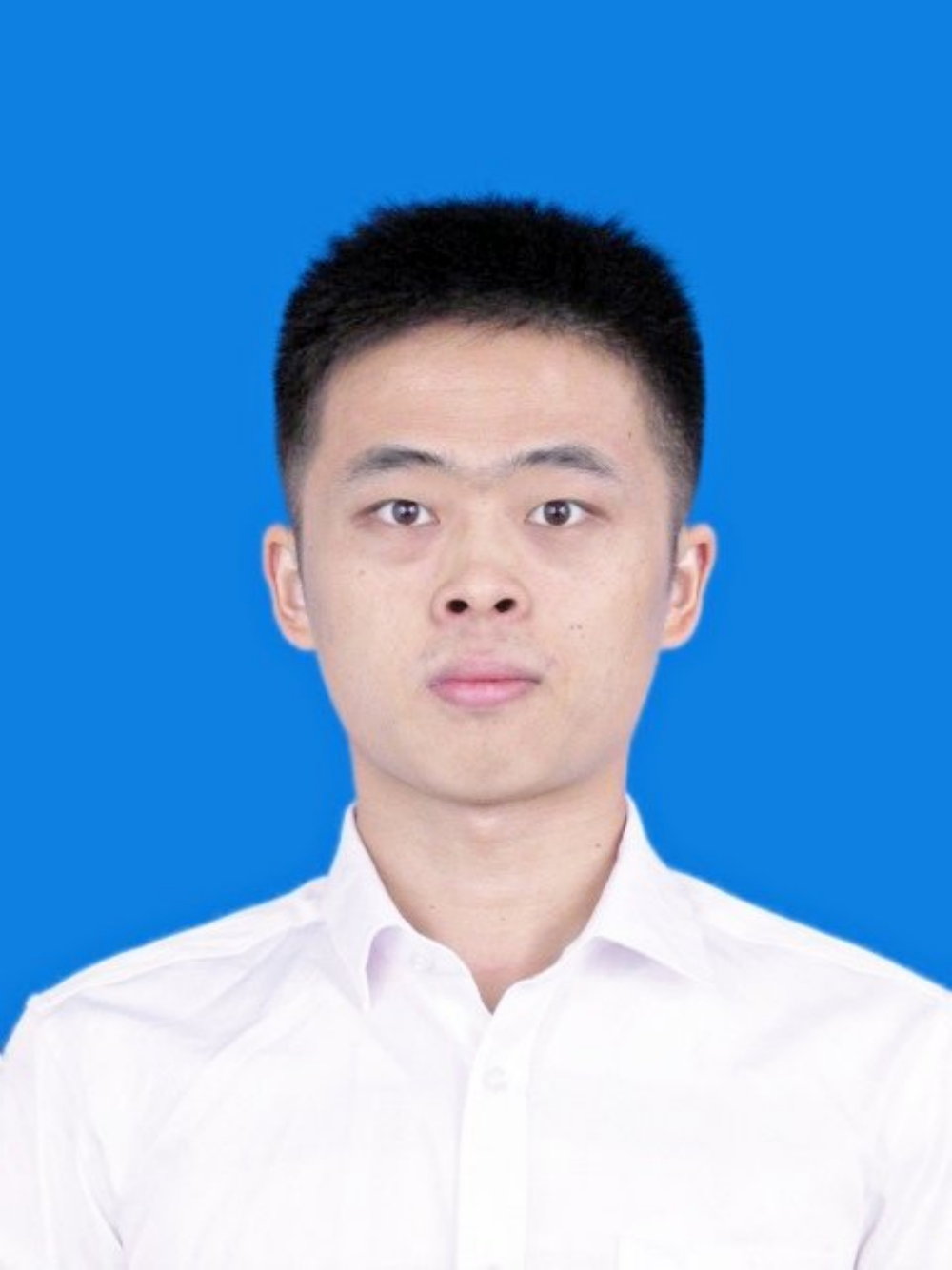}}]{Xiaoxiong Lu}
received B.Sc. degree in Applied Mathematics at the School of Science, Hunan University of Technology, in 2017.
He received M.Sc. degree in Computational Mathematics in 2020, and his PhD degree in Statistics in 2024, both from the School of Mathematics and Computational Science, Xiangtan University.
His research interests include complex networks, pseudo-random number and nonlinear dynamics.
\end{IEEEbiography}

\vskip -2\baselineskip plus -1fil

\begin{IEEEbiography}[{\includegraphics[width=1in,height=1.25in,clip,keepaspectratio]{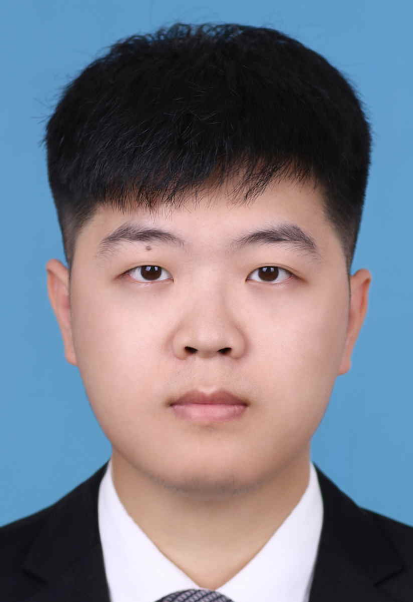}}]{Kai Tan}
received B.Sc. degree in mechanism design, manufacturing and automatization at the School of Mechanical Engineering, Xiangtan University in 2015.
He received his M.Sc. degree in computer science at the School of Computer Science, Xiangtan University in 2020.
Now, he is pursuing PhD degree at the same school.
His research interests include complex networks and nonlinear dynamics.
\end{IEEEbiography}

\vskip -2\baselineskip plus -1fil

\begin{IEEEbiography}[{\includegraphics[width=1in,height=1.25in,clip,keepaspectratio]{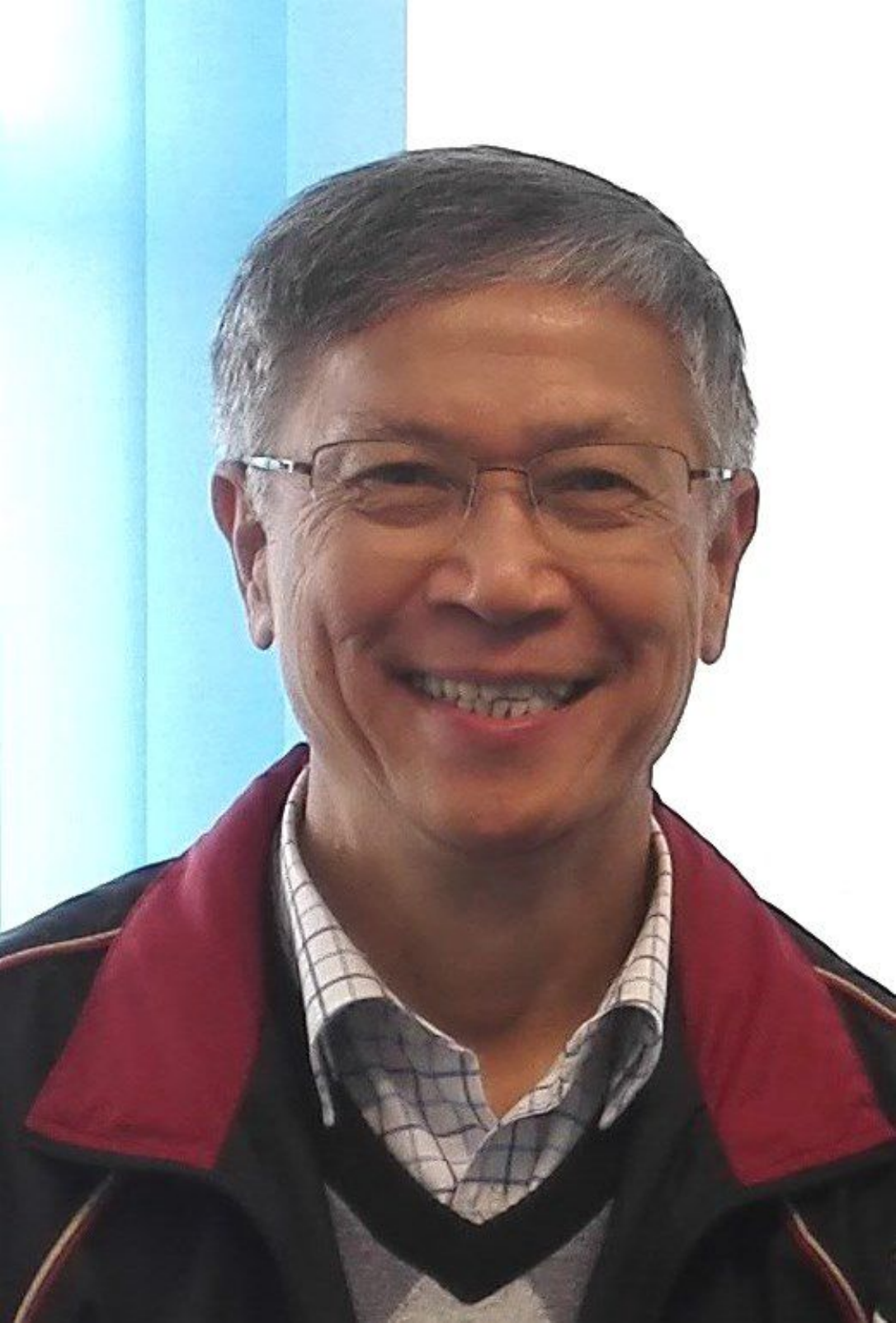}}]{Guanrong Chen}
(M'89--SM'92--F'97--LF'19) received the MSc. degree in Computer Science from Sun Yat-sen University, Guangzhou, China in 1981 and the Ph.D. degree in Applied Mathematics from Texas A\&M University, Texas, TX, USA in 1987. He has been a Chair Professor and the Director of the
Centre for Chaos and Complex Networks, City University of Hong Kong, Hong Kong since 2000, prior to that he was a tenured Full Professor with the University of Houston, Houston, TX, USA.

Prof. Chen was a recipient of the 2011 Euler Gold Medal, Russia, and Highly Cited Researcher Awards in Engineering as well as in Mathematics by Thomson Reuters, and conferred Honorary Doctorate by the Saint Petersburg State University, Russia in 2011 and by the University of Le Havre, Normandy, France in 2014. He is a member of the Academia of Europe since 2014 and a Fellow of The World Academy of Sciences since 2015.
\end{IEEEbiography}
\end{document}